\documentclass[11pt]{amsart}
\usepackage{amsmath,amssymb,times,epsfig,graphics}

\newtheorem{theorem}{{Theorem}}[section]
\newtheorem{proposition}[theorem]{{Proposition}}

\newtheorem{lemma}[theorem]{{Lemma}}

\newtheorem{corollary}[theorem]{{Corollary}}

\theoremstyle{definition}
\newtheorem{definition}[theorem]{{Definition}}

\theoremstyle{remark}
\newtheorem{remark}[theorem]{{Remark}}
\newtheorem{remarks}[theorem]{{Remarks}}

\newcommand{\NN}{\mathbb{N}}

\newcommand{\QQ}{\mathbb{Q}}
\newcommand{\RR}{\mathbb{R}}

\newcommand{\ZZ}{\mathbb{Z}}
\def\cA{{\mathcal A}}   \def\cM{{\mathcal M}} 
\def\cB{{\mathcal B}}    
   \def\cO{{\mathcal O}} \def\cU{{\mathcal U}}
    \def\cV{{\mathcal V}}
  \def\cK{{\mathcal K}}

\renewcommand{\epsilon}{\varepsilon}
\newcommand{\ba}{\overline{a}}

\newcommand{\1}{\mbox{\small{I}}}
\newcommand{\2}{\mbox{\small{II}}}
\newcommand{\6}{\mbox{\small{VI}}_0}
\newcommand{\7}{\mbox{\small{VII}}_0}
\newcommand{\8}{\mbox{\small{VIII}}}
\newcommand{\9}{\mbox{\small{IX}}}
\newcommand{\wC}{\widehat{C}}

\title[Aperiodic oscillatory asymptotic behavior for some Bianchi spacetimes]{Aperiodic oscillatory asymptotic behavior \\for some Bianchi spacetimes}

\author{Fran\c cois B\'eguin}
\address{Laboratoire de
Math\'ematiques (UMR 8628), Univ. Paris Sud 11, 91405 Orsay, France.}
\email{Francois.Beguin@math.u-psud.fr} 

\thanks{Work supported by ANR project GEODYCOS}


\begin{document}

\sloppy

\begin{abstract}
We study the asymptotic behavior of vacuum Bianchi type A spacetimes close to their singularity. It has been conjectured that this behavior is driven by a certain circle map, called the \emph{Kasner map}. As a step towards this conjecture, we prove that some orbits of the Kasner map do indeed attract some solutions of the system of ODEs which describes the behavior of vacuum Bianchi type A spacetimes. The orbits of the Kasner map for which we can prove such a result are those which are not periodic and do not accumulate on any periodic orbit. This shows the existence of Bianchi spacetimes  with an aperiodic oscillatory asymptotic behavior. 
\end{abstract}

\maketitle

\section{Introduction}

\subsection{Bianchi spacetimes, the Wainwright-Hsu vector field and the Kasner map} A \emph{Bianchi spacetime} is a cosmological spacetime which is spatially homogeneous.  More precisely, it is a spacetime $(M,g)$ such that $M$ is diffeomorphic to the product $G\times I$ of a simply connected three-dimensional Lie group $G$ and an interval $I$ of the real line, and  $g=h_t-dt^2$ where $h_t$ is a left-invariant riemannian metric on $G\times\{t\}\simeq G$ for every $t\in I$.  A \emph{vacuum} Bianchi spacetime is a Bianchi spacetime $(M,g)$ satisfying the vacuum Einstein equation $\mbox{Ric}_g=0$. A \emph{type A} Bianchi spacetime is a Bianchi spacetime for which the corresponding three dimensional Lie group $G$ is unimodular. Note that the Lie group $G$ can be assumed to be simply connected without loss of generality.

A Bianchi spacetime can be described as a one-parameter family $(h_t)_{t\in I}$ of left-invariant riemannian metrics on a three-dimensional Lie group $G$. The space of left-invariant riemannian metrics on a given Lie group is finite-dimensional. Therefore, when restricted to the context of Bianchi spacetimes, the vacuum Einstein equation becomes an ODE on a finite dimensional phase space. A classical and convenient way to write this ODE is to use the so-called \emph{Wainwright-Hsu variables} $(\Sigma_1,\Sigma_2,\Sigma_3,N_1,N_2,N_3)$. Very roughly, $N_1,N_2,N_3$ are the structure constants of the Lie algebra of the Lie group $G\times\{t\}$ in a certain basis, and $\Sigma_1,\Sigma_2,\Sigma_3$ are the components of the normalized traceless second fundamental form of $G\times\{t\}$ in the same basis. The corresponding phase space is the four-dimensional manifold
$$\cB:=\left \{(\Sigma_1,\Sigma_2,\Sigma_3,N_1,N_2,N_3)\in\RR^6 \mid \Sigma_1+\Sigma_2+\Sigma_3=0\,,\, \Omega=0 \right \} $$
 where 
$$\Omega = 6-(\Sigma_1^2+\Sigma_2^2+\Sigma_3^2)+\frac{1}{2}(N_1^2+N_2^2+N_3^2)-(N_1N_2+N_1N_3+N_2N_3).$$
The vacuum Einstein equation is equivalent to the system of ODEs
\begin{equation}
\label{e.Wainwright-Hsu}
\left\{
\begin{array}{rcl}
\Sigma_1' & = & (2-q) \Sigma_1 - R_1\\
\\
\Sigma_2' & = & (2-q) \Sigma_2 - R_2\\
\\
\Sigma_3' & = & (2-q) \Sigma_3 - R_3\\
\\
N_1' & = & -(q+2\Sigma_1)  N_1 \\
\\
N_2' & = & -(q+2\Sigma_2)  N_2\\
\\
N_3' & = & -(q+2\Sigma_3)  N_3
\end{array}
\right.
\end{equation}
where 
$$q  =  \frac{1}{3}\left(\Sigma_1^2+\Sigma_2^2+\Sigma_3^2\right)$$
and
$$\left\{
\begin{array}{rcl}
R_1 & = & \frac{1}{3} \left(2N_1^2-N_2^2-N_3^3+2N_2N_3-N_1N_3-N_1N_2\right)\\
\\
R_2 & = & \frac{1}{3} \left(2N_2^2-N_3^2-N_1^3+2N_3N_1-N_2N_1-N_2N_3\right)\\
\\
R_3 & = & \frac{1}{3} \left(2N_3^2-N_1^2-N_2^3+2N_1N_2-N_3N_2-N_3N_1\right)
\end{array}
\right.$$
(see e.g. \cite{HeinzleUggla2009} for the construction of the Wainwright-Hsu variables and the expression of the Einstein equation in these variables). In other words, vacuum Bianchi type A spacetimes can be seen as the solutions of  the system of ODEs~\eqref{e.Wainwright-Hsu} on the four-dimensional manifold $\cB$.  We will call \emph{Wainwright-Hsu vector field}, and denote by $X_{\cB}$, the vector field on $\cB$ associated to the system of ODEs~\eqref{e.Wainwright-Hsu}. We will denote by $X_{\cB}^t$ the time $t$ map of the flow of $X_{\cB}$. To measure  the distances on the phase space $\cB$, we will use the riemannian metric 
$h=(d\Sigma_1)^2+(d\Sigma_2)^2+(d\Sigma_3)^2+(dN_1)^2+(dN_2)^2+(dN_3)^2$. 

\begin{remark}
We have chosen the ``anti-physical" time orientation. With this convention, Bianchi spacetimes are future incomplete, but not necessarily past incomplete. The main reason for this choice is that we want the so-called \emph{mixmaster attractor} (see below) to be an attractor, rather than a repellor.
\end{remark}

The phase space $\cB$ admits a natural stratification, which is invariant under the flow of the Wainwright-Hsu vector field $X_{\cB}$. There are six strata denoted by $\cB_{\1},\cB_{\2},\cB_{\6},\cB_{\7},\cB_{\8},\cB_{\9}$. These strata correspond to the different possible signs for the variables $N_1$, $N_2$, $N_3$. They also correspond to the different (isomorphism class of) simply connected unimodular three-dimensional Lie groups. The orbits of $X_{\cB}$ contained in $\cB_{\1}$ are called type $\1$ orbits, the orbits contained in $\cB_{\2}$ are called type $\2$ orbits, etc. The behavior of the type $\1$, $\2$, $\6$ and $\7$ orbits of $X_{\cB}$ is very well understood. On the contrary, the behavior of the type $\8$ and $\9$ orbits (which are the generic orbits in $\cB$) is, at best, conjectural. In order to simplify the discussion we will focus our attention on the subset $\cB^+$ of $\cB$ where the coordinates $N_1,N_2,N_3$ are non-negative:
$$\cB^+:=\{(\Sigma_1,\Sigma_2,\Sigma_3,N_1,N_2,N_3)\in\cB \mid N_1\geq 0\;,\;N_2\geq 0\;,\;N_3\geq 0\}.$$
Observe that this subset is invariant under the flow of $X_{\cB}$. 

The stratum $\cB_{\1}$ is more frequently denoted by $\cK$. This is a euclidean circle in $\cB\subset\RR^6$, called the \emph{Kasner circle}. It is made of the points of $\cB$ where $N_1=N_2=N_3=0$. This corresponds to the case where the Lie group $G$ is abelian (\emph{i.e.} $G=\RR^3$).
\begin{eqnarray}
\label{e.Kasner-circle}
\quad \cK & = & \{(\Sigma_1,\Sigma_2,\Sigma_3,N_1,N_2,N_3)\in\cB \mid N_1=N_2=N_3=0\}\\
\nonumber & = & \{(\Sigma_1,\Sigma_2,\Sigma_3,N_1,N_2,N_3)\in\RR^6 \mid N_1=N_2=N_3=0\,, \\
\nonumber & & \quad\quad\quad\quad\quad\quad\quad\quad\quad\quad\quad\quad\quad\quad   \Sigma_1+\Sigma_2+\Sigma_3=0 \;,\; \Sigma_1^2+\Sigma_2^2+\Sigma_3^2=6\}
\end{eqnarray}
The vector field $X_{\cB}$ vanishes on $\cK$. Therefore every type $\1$ orbit is a fixed point. There are three points on $\cK$, called the \emph{special points} or the \emph{Taub points}, which play a very important role in every attempt to understand the behavior of the solution of the Wainwright-Hsu equations; these are the points for which $(\Sigma_1,\Sigma_2,\Sigma_3)$ is equal respectively to $(2,-1,-1)$, $(-1,2,-1)$ and $(-1,-1,2)$; these points are usually denoted by $T_1,T_2,T_3$. For every point $p\in\cK$, the derivative $DX_{\cB}(p)$ has four distinct eigendirections: the direction $T_p\cK$ of the Kanser circle corresponds to a zero eigenvalue (since $X_{\cB}$ vanishes on~$\cK$), and the directions $\RR.\frac{\partial}{\partial N_1}(p)$, $\RR.\frac{\partial}{\partial N_2}(p)$, $\RR.\frac{\partial}{\partial N_3}(p)$ are respectively associated to the eigenvalues $-(2+2\Sigma_1)$, $-(2+2\Sigma_2)$, $-(2+2\Sigma_3)$.  When $p$ is one of the three special points $T_1,T_2,T_3$, two of the three eigenvalues $-(2+2\Sigma_1),-(2+2\Sigma_2),-(2+2\Sigma_3)$ vanish; hence, the derivative of $DX_{\cB}(p)$ has a triple zero eigenvalue. If $p$ is not one of the three special points $T_1,T_2,T_3$, then the eigenvalues $-(2+2\Sigma_1),-(2+2\Sigma_2),-(2+2\Sigma_3)$ are pairwise distinct, two of them are negative, and the other one is positive; in other words, the derivative $DX_{\cB}(p)$ has two distinct negative eigenvalue, a multiplicity one eigenvalue, and a positive eigenvalue. 

The stratum $\cB_{\2}$ is two-dimensional. It is made of the points of $\cB$ for which exactly two of the $N_i$'s vanish. This corresponds to the case where the Lie group $G$ is isomorphic to the Heisenberg group.
$$\begin{array}{rcl}
\cB_{1}\cup\cB_{\2} & = & \{(\Sigma_1,\Sigma_2,\Sigma_3,N_1,N_2,N_3)\in\cB \mid N_1=N_2=0 \mbox{ or }N_1=N_3=0\mbox{ or }N_2=N_3=0\}
\end{array}$$
One can easily check that $\cB_{\1}\cup\cB_{\2}$ is the union of three ellipsoids which intersect along the Kasner circle. The  vector field $X_{\cB}$ commutes with the transformation $(N_1,N_2,N_3)\mapsto (-N_1,-N_2,-N_3)$, so we can work in $\cB_{\2}\cap\cB^+$ to simplify the discussion. The intersection  $\cB_{\2}\cap\cB^+$ is the union of three disjoint hemi-ellipsoids bounded by the Kasner circle. Type $\2$ orbits can be calculated explicitly (see for example~\cite{HeinzleUggla2009}). It appears that every type $\2$ is a heteroclinic orbit connecting a point of $\cK\setminus\{T_1,T_2,T_3\}$ to another point of $\cK\setminus\{T_1,T_2,T_3\}$. Moreover, for every point $p\in \cK\setminus\{T_1,T_2,T_3\}$, there is exactly one type $\2$ orbit in $\cB^+$ which ``takes off" at $p$ (this orbit is tangent to the eigendirection associated to the positive eigenvalue of $DX_{\cB}(p)$) and two type $\2$ orbits in $\cB^+$ which ``land" at $p$ (these orbits are tangent to the eigendirections associated to the negative eigenvalues of $DX_{\cB}(p)$). 

The Kasner map $f:\cK\to\cK$ is defined as follows. Consider a point $p\in\cK$. If $p$ is one of the three special points $T_1,T_2,T_3$, then $f(p)=p$. Otherwise, one considers the unique type $\2$ orbit which ``springs up" at $p$; this orbit connects $p$ to another point $p'\in\cK$; this point $p'$ is by definition the image of $p$ under $f$. For every $p\in\cK\setminus\{T_1,T_2,T_3\}$, we will denote by $\cO_{p,f(p)}$ the unique type $\2$ orbit connecting the point $p$ to the point $f(p)$ in $\cB^+$. Since there are exactly two type $\2$ orbits in $\cB^+$ which land at a given point $p\in\cK$, the Kasner map is two-to-one. The exact computation of the type $\2$ orbits of $X_\cB$ yields a nice geometric description of the Kasner map $f$. Consider the equilateral triangle which is tangent to $\cK$ at the three special points $T_1,T_2,T_3$. Denote by $M_1,M_2,M_3$ the vertices of this triangle. For $p\in\cK\setminus\{T_1,T_2,T_3\}$, let $M_i$ be the vertex of $\cK$ which is the closest to~$p$. The  line $(p\,M_i)$ and intersects the circle $\cK$ at two points: the point $p$ and the point $f(p)$. See figure~\ref{f.Kasner-map}. Using this geometric description, it is easy to see $f$ is $C^1$. Moreover, the Kasner map $f$ is non-uniformly expanding. Indeed, the norm of the derivative of $f$ (calculated with respect to the metric on $\cK$ induced by the riemannian metric $h$) at a point $p\in\cK$ is equal to $1$ if $p$ is one of the three special points $T_1,T_2,T_3$, and is strictly bigger than $1$ if $p$ is not one of these three special points. It follows that, for every compact set $C$ in $\cK\setminus\{T_1,T_2,T_3\}$ there exists a constant $\nu_C>1$ such that $||| Df(p)|||_h\geq \nu_C$ for every $p$ in $C$. 

\begin{figure}
\centerline{\includegraphics[height=6cm]{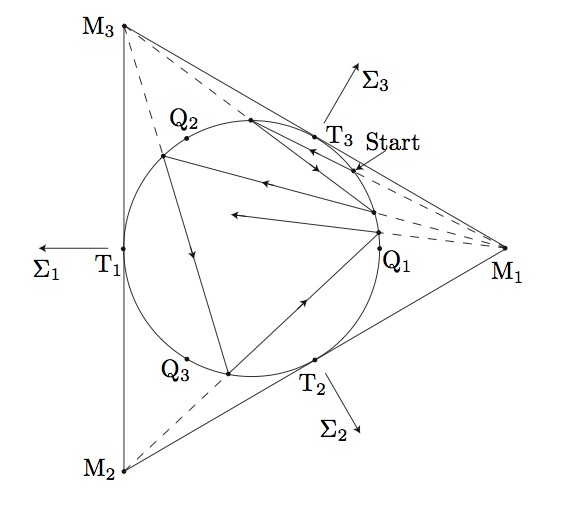}}
\caption{\label{f.Kasner-map}Construction of an orbit of the Kasner map}
\end{figure}

The strata $\cB_{\6}$ and $\cB_{\7}$ are three-dimensional. The stratum $\cB_{\6}$ is made of the points of $\cB$ for which exactly one of the $N_i$'s is equal to zero, the two others being of different signs. This corresponds to the case where the Lie group $G$ is isomorphic to $\mbox{Isom}(\RR^{1,1})=\mbox{O}(1,1)\ltimes\RR^2$. This stratum is disjoint from $\cB^+$. The stratum $\cB_{\7}$ is made of the points of $\cB$ for which exactly one of the $N_i$'s is equal to zero, the two others being of the same sign. This corresponds to the case where the Lie group $G$ is isomorphic to $\mbox{Isom}(\RR^2)=O(2)\ltimes\RR^2$.

Finally the strata $\cB_{\8}$ and $\cB_{\9}$ are four-dimensional (these are open subsets of $\cB$). The stratum $\cB_{\8}$ is made of the points of $\cB$ for which the $N_i$'s are non-zero and do not have the same sign. It corresponds  to the case where the Lie group $G$ is isomorphic to the universal cover of $\mbox{SL}(2,\RR)$. It is disjoint from $\cB^+$. The stratum $\cB_{\9}$ is made of the points of $\cB$ for which the $N_i$'s are either all positive, or all negative. This corresponds  to the case where the Lie group $G$ is isomorphic to $\mbox{SO}(3,\RR)$.

\subsection{Statement of the main results of the paper}
Misner has conjectured that the dynamics of type $\9$ orbits of the Wainwright-Hsu vector field $X_{\cB}$ ``is driven" by the Kasner map (\cite{Misner1969}). The idea is that every type $\9$ orbit should eventually approach the so-called \emph{mixmaster attractor} $\cA:=\cK\cup\cB_{\2}$, and then ``follow" a heteroclinic chain \hbox{$q\to \cO_{q,f(q)}\to f(q)\to \cO_{f(q),f^2(q)}\to\dots$}. Conversely, every heteroclinic chain as above should attract some type $\9$ orbits of $X_{\cB}$. Moreover,   ``generic" type $\9$ orbits should be attracted by ``generic" heteroclinic chains. This would imply that the behavior of generic type $\9$ orbits of $X_{\cB}$ is determined by the behavior of generic orbits of the Kasner map. See~\cite{HeinzleUggla2009} for a detailed discussion of various possible precise  statements for this conjecture.

In 2001, Ringstr\"om proved that $\cA:=\cK\cup\cB_{\2}$ is indeed an global attractor: the distance from almost every type~$\9$ orbit of the Wainwright-Hsu vector field $X_{\cB}$ to $\cA$ tends to $0$ as the time goes to $\infty$ (\cite{Ringstrom2001}, see also~\cite{HeinzleUggla2009bis}). Ringstr\"om's theorem has important consequences, such as the divergence of the curvature in Bianchi spacetimes as one approaches their past (or future) singularity. Nevertheless, this theorem does not tell anything on the relation between the dynamics of Bianchi orbits and the dynamics the Kasner map. The purpose of the present paper is to prove that ``there are many points $q\in\cK$ such that the heteroclinic chain $q\to \cO_{q,f(q)}\to f(q)\to \cO_{f(q),f^2(q)}\to\dots$ does attract some type $\9$ orbits of $X_{\cB}$".

\begin{definition}
\label{d.stable-manifold}
For $q\in\cK$, we call \emph{stable manifold of $q$}, and we denote by $W^s(q)$, the set of all points $r\in\cB$ for which there exists an increasing sequence of real numbers $(t_n)_{n\geq 0}$ such that 
$$\mbox{dist}(X_{\cB}^{t_n}(r),f^n(q))\mathop{\longrightarrow}_{n\to \infty} 0$$
and such that the Hausdorff distance between the piece of orbit $\{X_{\cB}^t(r) \; ;\; t_n\leq t\leq t_{n+1}\}$ and the type$\2$ heteroclinic orbit $\cO_{f^n(q),f^{n+1}(q)}$  tends to $0$ as $n$ goes to $+\infty$.
\end{definition}

A subset $C$ of $\cK$ is said to be \emph{ forward-invariant} if $f(C)\subset C$. It is said to be \emph{aperiodic} if it does not contain any periodic orbit for $f$.  Our main result can be stated as follows: 
 
\begin{theorem}
\label{t.main}
If $q\in\cK$ is contained in a closed  forward-invariant aperiodic subset of $\cK$, then the intersection of the stable manifold $W^s(q)$ with $\cB_{\9}$ is non-empty. More precisely, $W^s(q)\cap\cB_{\9}$ contains a $C^1$ embedded three-dimensional closed disc. This disc can be chosen to depend continuously on $q$ (for the $C^1$ topology on the space of $C^1$ embeddings of the closed unit disc of $\RR^3$ in $\cB$) when $q$ ranges in a closed  forward-invariant aperiodic subset of $\cK$. 
\end{theorem}

\begin{remark}
\label{r.global-stable-manifold}
Actually, one can prove that $W^s(q)\cap\cB_{\9}$ is an injectively immersed open disc, and that this depends continuously on $q$ (for the compact-open $C^1$ topology on  the space of $C^1$ immersions of the open unit disc of $\RR^3$ in $\cB_{\9}$) when $q$ ranges in a closed  forward-invariant aperiodic subset of $\cK$. Note that $W^s(q)\cap\cB_{\9}$ is never a properly embedded open disc.
\end{remark}

\begin{remark}
We decided to focus on type $\9$ orbits for sake of simplicity. Nevertheless, the analog of theorem~\ref{t.main} for $\cB_{\8}$ instead of $\cB_{\9}$ is also true. The proof is exactly the same: one just needs to replace the set $\cB^+$ by the set 
$\cB^{(-,+,+)}:=\cB\cap \{N_1\leq 0\}\cap \{N_2\geq 0\}\cap \{N_3\geq 0\}.$
\end{remark}

Observe that the hypothesis of theorem~\ref{t.main} is satisfied by a dense set of points in $\cK$~:

\begin{proposition}
\label{p.existence-aperiodic}
The union of all the closed  forward-invariant aperiodic subsets of $\cK$ is in dense in~$\cK$. 
\end{proposition}

\begin{remark}
Let $E$ be  the union of all the closed  forward-invariant aperiodic subsets of $\cK$. According to the above proposition, $E$ is dense in $\cK$. Observe nevertheless that $E$ is a ``small" subset of $\cK$ both from the topological viewpoint (it is a meager set) and from the measurable viewpoint (it has zero Lebesgue measure). Also observe that theorem~\ref{t.main} does not tell that $W^s(q)$ depends continuously on $q$ when $q$ ranges in $E$. So, we do not know whether the union of all the stable manifolds $W^s(q)$, where $q$ ranges $E$, is dense in $\cB$ (or in an open subset of $\cB$) or not. 
\end{remark}

The origin of the hypothesis of theorem~\ref{t.main} is purely technical: if $p$ belongs to a  closed forward-invariant aperiodic subset of $\cK$, we can find a  coordinate system in which the vector field $X_{\cB}$ depends linearly of all the coordinates but one; this makes the estimates on the flow of $X_{\cB}$ near such a point much easier. We do not know if such linearizing coordinate systems exists near a point of $\cK$ which is periodic or preperiodic under the Kasner map. Nevertheless, it should be noticed that the estimates on the flow of $X_{\cB}$ that we need to construct stable manifolds are much weaker than the existence of a linearizing coordinate system. So, we do think that theorem~\ref{t.main} can be extended to any point $q\in\cK$ such that the orbit of $q$ under the Kasner map does not accumulate one any of the three special points $T_1,T_2,T_3$. Understanding the behavior of the orbits which pass arbitrary close to the three special points seems to be a much harder problem.

\begin{remark}
A consequence of theorem~\ref{t.main} (and proposition~\ref{p.existence-aperiodic}) is that the Wainwright-Hsu vector field is, at least partially, sensitive to initial conditions: there is a dense subset $E$ of the Kasner circle $\cK$ such that, for every point $q\in E$, arbitrarily closed to $q$, one can find two points $r_1,r_2\in\cB_{\9}$ such that the orbits of $r_1$ and $r_2$ will not have the same ``asymptotic behavior" (for example, $\displaystyle\mathop{\limsup}_{t\to +\infty} \mbox{dist}_h\left(X_{\cB}^t(r_1)\,,\,X_{\cB}^t(r_2)\right)>1/10$.
\end{remark}

\bigskip

 As we were finishing to write the present paper, M. Georgi, J. H\"arterich, S. Liebscher and K. Webster put on arXiv a preprint in which they prove that the period three orbits of $f$ admit non-trivial stable manifolds (\cite{GHLW2010}). At the end of this preprint, they claim that their techniques can be used to extend their result to any periodic orbit of $f$, and even to any orbit of $f$ which does not accumulate one any of the three special points $T_1,T_2,T_3$. Such a extension would imply our result.
 
\subsection{Idea of the proof and organization of the paper}

Let us try to sketch the key idea of the proof of our main theorem~\ref{t.main}. Let $C$ be a closed  forward-invariant aperiodic subset of the Kasner circle, and denote by $\wC$ the union of $C$ and all the type $\2$ orbits connecting two points of $C$ in $\cB^+$. We will construct a kind of Poincar\'e section adapted to $\wC$: a hypersurface with boundary $M_C$ which intersects transversally every type $\2$ orbit connecting  two points of $C$ in $\cB^+$. Theorem~\ref{t.main} will follow from the existence of non-trivial local stable manifolds for the Poincar\'e map $\Phi_C:M_C\to M_C$. We will use a classical stable manifold theorem for hyperbolic compact set; so, we will be left to prove that the compact set $\wC\cap M_C$ is a hyperbolic set for the Poincar\'e map $\Phi_C$. The key point is to understand the behavior of the orbits of the Wainwright-Hsu vector field $X_{\cB}$ close to a point $p\in C$. Roughly speaking, we will prove the following: when an orbit of $X_{\cB}$ passes close to a point $p\in C$, the distance of this orbit to the mixmaster attractor $\cA=\cK\cup\cB_{\2}$ decreases super-linearly, while the drift of this orbit tangentially to $\cA$ is as small as wanted. We do not have any precise control of what happens to the orbits of $X_{\cB}$ far from the Kasner circle, but we do not need to. Indeed, the duration of the excursion of the orbits of $X_{\cB}$ outside any given neighborhood of $\cK$ is universally bounded. A consequence is that everything that happens far from $\cK$ is dominated by the super-linear contraction of the distance to $\cA$ that occurs when an orbit passes close to $C$. This will be enough to obtain the desired hyperbolicity result. As already explained, the reason why we need to consider an \emph{aperiodic} subset of $\cK$ is purely technical: close to a point of $\cK$ which is not preperiodic under the Kasner map, we have a ``nice" coordinate system which is convenient to study precisely the behavior of the orbits of $X_{\cB}$. 

Let us now explain the organization of the paper. Consider a vector field $X$ on some manifold and a point $p$ such that $X(p)=0$. We say that $X$ satisfies \emph{Sternberg's condition} at $p$ if the non purely imaginary eigenvalues of $DX(p)$, counted with multiplicities, are linearly independent over $\QQ$.  Takens has proved a generalization of Sternberg's theorem which states that, if $X$  satisfies \emph{Sternberg's condition} at $p$, there exists a $C^1$ local ``linearizing" coordinate system for $X$ on a neighborhood of $p$. A precise statement of this theorem will be given in section~\ref{s.Takens}. The purpose of section~\ref{s.linearisation} is to apply Takens' theorem to the Wainwright-Hsu vector field. It will be proved that the Wainwright-Hsu vector field $X_{\cB}$ satisfies Sternberg's condition at some point $p$ of the Kasner circle if and only if $p$ is not preperiodic under the Kasner map. In section~\ref{s.Dulac}, the ``linearizing" coordinate system provided by Takens' theorem will be used to study the behavior of the orbits of the Wainwright-Hsu vector field close to a point of the Kasner circle which is not periodic under the Kasner map. Now consider a closed  forward-invariant aperiodic subset $C$ of the Kasner circle $\cK$ is considered, and denote by $\wC$ the union of $C$ and all the type $\2$ orbits connecting two points of $C$ in $\cB^+$. A Poincar\'e section $M_C$ adapted to $\wC$ will be constructed in section~\ref{s.construction-Poincare-section}; we will denote by $\Phi_C$ the corresponding Poincar\'e map. The existence of non-trivial stable manifolds for the corresponding Poincar\'e map $\Phi_C:M_C\to M_C$ will be proved in section~\ref{s.stable-manifolds-Poincare-map}, using the results of section~\ref{s.Dulac}. The proof of theorem~\ref{t.main} will be completed in section~\ref{s.end-proof-main}. Finally, the last section of the paper is devoted to the proof of proposition~\ref{p.existence-aperiodic}.

\subsection*{Acknowledgements}
I would like to thank Lars Andersson for some stimulating discussion, and for his encouragements to write the present paper.


\section{Takens' linearization theorem}
\label{s.Takens}

Let $X$ be a $C^\infty$ vector field on some manifold $\cM$, and $p$ be a point in $\cM$ such that $X(p)=0$. The linear space $T_p\cM$ admits a unique decomposition 
$$T_p\cM = E^s\oplus E^c\oplus E^u$$
where $E^s$, $E^c$, $E^u$ are $DX(p)$-invariant linear subspaces, the eigenvalues of $DX(p)_{|E^s}$ have negative real parts, the eigenvalues of $DX(p)_{|E^c}$ are purely imaginary, and the eigenvalues of $DX(p)_{|E^u}$ have positive positive. We denote by $s$, $c$ and $u$ the dimensions of the linear subspaces $E^s$, $E^c$ and $E^u$.

\begin{definition}
The vector field $X$ satisfies \emph{Sternberg's condition} at $p$ if the eigenvalues of $DX(p)_{|E^s\oplus E^u}$ (counted with multiplicities) are linearly independent over $\QQ$.
\end{definition}

F. Takens as proved the following generalization of the classical Sternberg linearization theorem:

\begin{theorem}[see~\cite{Takens1971}, page 144]
\label{t.Takens}
Assume that $X$ satisfies Sternberg's condition at $p$. Then, for every $r\geq 0$, one can find a neighborhood $U$ and a $C^r$ coordinate system $(x_1,\dots,x_s,y_1,\dots,y_c,z_1,\dots,z_u)$ on $U$ centered at $p$, such that, in this coordinate system, $X$ reads:
\begin{equation}
\label{e.Takens}
X =  \sum_{i,j=1}^{s} a_{i,j}(y_1,\dots,y_c) x_j\frac{\partial}{\partial x_i} + \sum_{i=1}^{c}\phi_i(y_1,\dots,y_c)\frac{\partial}{\partial y_i} +  \sum_{i,j=1}^{u} b_{i,j}(y_1,\dots,y_c) z_j\frac{\partial}{\partial z_i} 
\end{equation}
where the eigenvalues of  the matrix $(a_{i,j}(0,\dots,0))_{1\leq i,j\leq s}$ have negative real parts, the eigenvalues of the matrix $\big(\frac{\partial \phi_i}{\partial y_j}(0,\dots,0)\big)_{1\leq i,j\leq c}$ are purely imaginary,
and  the eigenvalues of  the matrix $(b_{i,j}(0,\dots,0))_{1\leq i,j\leq u}$ have positive real parts.
\end{theorem} 

Note that, in general, the size of the neighborhood $U$ does depend on the integer $r$, and it is not possible to find any $C^\infty$ local coordinate system centered at $p$ such that~\eqref{e.Takens} holds. In this sense, Takens' theorem is not a true generalization of Sternberg's theorem.

Also observe that the name "Takens' linearization theorem" is slightly incorrect: indeed, the vector field $X$ is not linear in the $(x_1,\dots,x_s,y_1,\dots,y_c,z_1,\dots,z_u)$ coordinate system. Nevertheless, $X$ depends linearly on the coordinates $x_1,\dots,x_s$ and $z_1,\dots,z_u$. Also note that the submanifold defined by the equation $(y_1,\dots,y_c)=(0,\dots,0)$ is invariant under the flow of $X$, and that the restriction of $X$ to this submanifold is linear. For $(\zeta_1,\dots,\zeta_c)\neq (0,\dots,0)$,  the submanifold defined by the equation $(y_1,\dots,y_c)=(\zeta_1,\dots,\zeta_c)$ is not invariant under the flow of $X$, but the projection of $X$ on the tangent space of this submanifold is linear.

Of course, equality~\eqref{e.Takens} together with the signs of the real parts of the eigenvalues of the matrices $(a_{i,j}(0,\dots,0))_{1\leq i,j\leq s}$, $\big(\frac{\partial \phi_i}{\partial y_j}(0,\dots,0)\big)_{1\leq i,j\leq c}$ and $(b_{i,j}(0,\dots,0))_{1\leq i,j\leq u}$ implies that:
\begin{itemize}
\item[--] the vectors $\frac{\partial}{\partial x_1}(p),\dots,\frac{\partial}{\partial x_s}(p)$ span the linear subspace $E^s$, 
\item[--] the vectors $\frac{\partial}{\partial y_1}(p),\dots,\frac{\partial}{\partial y_c}(p)$ span the linear subspace $E^c$,
 \item[--] the vectors $\frac{\partial}{\partial z_1}(p),\dots,\frac{\partial}{\partial z_u}(p)$ span the linear subspace $E^u$.
 \end{itemize}


\section[Linearization of Wainwright-Hsu vector field]{Linearization of Wainwright-Hsu vector field near a point of the Kasner circle which is not preperiodic under the Kasner map}
\label{s.linearisation}

The purpose of this section is to apply Takens' theorem to the Wainwright-Hsu vector field at a point of the Kasner circle. For this purpose, we will need to characterize the points $p$ on the Kasner circle such that the Wainwirght-Hsu vector field $X_{\cB}$ satisfies Sternberg's condition at $p$. We will see that these are exactly the points $p\in\cK$ which are not preperiodic under the Kasner map $f$. In order to relate the arithmetic properties of the eigenvalues of the derivative of $DX_{\cB}(p)$ and the behavior of the orbit of $p$ under $f$, we will use the so-called \emph{Kasner parameter}.

\subsection{Kasner parameter}
Let $q=(\Sigma_1,\Sigma_2,\Sigma_3,0,0,0)$ be a point of the Kasner circle. The \emph{Kasner parameter} of $q$ is the unique real number $u=u(q)\in [1,\infty]$ which satisfies the following equality:
\begin{equation}
\label{e.Kasner-parameter}
\big(2+\Sigma_1\big)\big(2+\Sigma_2\big)\big(2+\Sigma_3\big)=\frac{-216\;u^2\big(1+u\big)^2}{\big(1+u+u^2\big)^3}.
\end{equation}
The map $q\mapsto u(q)$ is not one-to-one. Nevertheless, the point $q$ is characterized by its Kasner parameter up to permutations of the coordinates $\Sigma_1,\Sigma_2,\Sigma_3$. More precisely, equality~\eqref{e.Kasner-parameter} together with the equation of the Kasner circle~\eqref{e.Kasner-circle} imply that:
\begin{equation}
\label{e.eigenvalues-Kasner-parameter}
\Big(2+\Sigma_1\,,\, 2+2\Sigma_2\,,\, 2+2\Sigma_3\Big)\; \mathop{\mbox{=}}^{\mbox{up to}}_{\mbox{permutation}}\;
 \left(\frac{-6u}{1+u+u^2}\,,\, \frac{6(1+u)}{1+u+u^2} \,,\, \frac{6u(1+u)}{1+u+u^2}\right).
\end{equation}
Note that $u=\infty$ if and only if $q$ is one of the three special points $T_1,T_2,T_3$. The main advantage of the Kasner parameter is the fact that the Kasner map $f$ admits a nice expression in terms of this parameter: for every $q\in\cK$, one has $u(f(q))=\bar f(u(q))$ where $\bar f:[1,+\infty]\to [1,\infty]$ is defined by
\begin{equation}
\label{e.Kasner-map-Kasner-parameter}
\bar f(u)=\left\{
\begin{array}{ll}
\infty & \mbox{ if } u=1 \mbox{ or }\infty \\\\
u-1 & \mbox{ if } u\geq 2\\\\
\frac{1}{u-1} & \mbox{ if } 1<u\leq 2
\end{array}\right.
\end{equation}
(see, for example,~\cite{HeinzleUggla2009}).

\subsection{Characterization of the points of the Kasner circle where Sternberg's condition is satisfied}
The proposition below gives a necessary and sufficient condition for the Wainwright-Hsu vector field~$X_{\cB}$ to satisfy Sternberg's condition at a point $p\in\cK$, in terms of the behavior of the orbit of $p$ under the Kasner map. The hypothesis of our main theorem~\ref{t.main} comes directly from this condition.

\begin{proposition}
\label{p.characterisation-Sternberg}
Let $p$ be a point of the Kasner circle $\cK$ which is not one of the three special points~$T_1,T_2,T_3$. The three following conditions are equivalent~:
\begin{enumerate}
\item the vector field $X_{\cB}$ satisfies Sternberg's condition at $p$~; 
\item the Kasner parameter $u(p)$ is neither a rational number, nor a quadratic irrational number~;
\item the orbit of $p$ under the Kasner map $f$ is not preperiodic. 
\end{enumerate}
\end{proposition}

\begin{proof}
Denote by $(\Sigma_1,\Sigma_2,\Sigma_3,0,0,0)$ the coordinates of $p$. Since $p$ is not one of the three special points, the derivative $DX_{\cB}(p)$ has two distinct negative eigenvalues, one zero eigenvalue, and one positive eigenvalue. The three non-zero eigenvalues of $DX_{\cB}(p)$ are equal to $-(2+\Sigma_1)$, $-(2+\Sigma_2)$ and $-(2+\Sigma_3)$.

Let us prove the equivalence between $(1)$ and $(2)$. The vector field $X_{\cB}$ satisfies Sternberg's condition at $p$ if and only if the real numbers $-(2+\Sigma_1)$, $-(2+\Sigma_2)$ and $-(2+\Sigma_3)$ are linearly independent over~$\QQ$. Using formula~\eqref{e.eigenvalues-Kasner-parameter}, one sees that this is equivalent to the fact that the real numbers $u(p)$, $1+u(p)$ and $u(p)(1+u(p))$ are independent over $\QQ$. Clearly, this is equivalent to the fact the real number $u(p)$ is neither a rational number, nor a quadratic irrational number.

Now, let us prove the equivalence between $(2)$ and $(3)$. Recall that, for every $q$ on the Kasner circle, one has $u(f(q))=\bar f(u(q))$ where $\bar f:[1,+\infty]\to [1,\infty]$ is given by~\eqref{e.Kasner-map-Kasner-parameter}. Observe that both the set of rational numbers and the set of irrational numbers are invariant under $\bar f$. So we can treat the case where $u(p)$ is rational and the case where $u(p)$ is irrational separately. 

First consider the case where $u(p)$ is rational. Then it is very easy to prove that the orbit of $u(p)$ under $\bar f$ ``ends up" at $\infty$. Now recall that $u(q)=\infty$ if and only if $q$ is one of the three special points. Hence the orbit of $p$ under $f$ ``ends up" at  one of the three special points. In particular, the orbit of $p$  is preperiodic. 

Now consider the case where $u(p)$ is irrational. 
Looking again at~\eqref{e.Kasner-map-Kasner-parameter}, one sees that the orbit of $u(p)$ under $\bar f$ returns an infinite number of times in the interval $(1,2]$. Let $k:(1,\infty)\setminus\QQ\to\NN\setminus\{0\}$ be the return time function of $\bar f$ in the interval $(1,2]$, and $\overline F:(1,2]\setminus\QQ\toÊ(1,2]\setminus\QQ$ be the first return map of $\bar f$ in the interval $(1,2]$, \emph{that is} 
$$k(u)=\inf \{n>0\mbox{ such that }\bar f^n(u)\in (1,2]\}\quad \mbox{and}\quad \overline F(u)=\bar f^{k(u)}(u).$$ 
Then, 
$$k(u)=\left\lfloor\frac{1}{u-1}\right\rfloor -1\quad \mbox{and}\quad \overline F(u)=\left\{\frac{1}{u-1}\right\}+1,$$
where $\lfloor x\rfloor$ is the integer part of $x$ and $\{x\}:=x-\lfloor x\rfloor$ is the fractional part of $x$.
The point $p$ is preperiodic under the Kasner map $f$ if and only if either $u(p)$ is a preperiodic under $\bar f$, \emph{that is} if and only if $\bar f^{k(u(p))}(u(p))$ is preperiodic under $\overline F$. Now observe that $\overline F$ is just the Gauss map $u\mapsto \left\{\frac{1}{u}\right\}$ conjugated by the translation $u\mapsto u+1$, and that $k(u-1)+1$ is the first term of the continuous development fraction of $u$. On the one hand, the preperiodic points of the Gauss map $u\mapsto \left\{\frac{1}{u}\right\}$ are exactly the real numbers $u\in (0,1]$ such that the sequence of integers which appear in the continuous fraction development of $u$ is preperiodic. On the other hand, it is well-known that the continuous fraction development of $u\in\RR$ is preperiodic if and only if $u$ is a quadratic irrational number. This shows that the orbit of $p$ is a preperiodic under the Kasner map $f$ if and only if $u(p)$ is a quadratic irrational number.
\end{proof}

\subsection{Linearization of the Wainwright-Hsu vector field}
According to proposition~\ref{p.characterisation-Sternberg}, if $p$ is not a preperiodic point for the Kasner map, the hypotheses of Takens linearization theorem~\ref{t.Takens} are satisfied by $X_{\cB}$ at $p$. This theorem provides us with a local coordinate system on a neighborhood of $p$ in which $X_{\cB}$ is ``almost linear":

\begin{proposition}
\label{p.linearisation}
Let $p$ be a point of the Kasner circle $\cK$ which is not preperiodic under the Kasner map (in particular, $p$ is not a special point). Then there exists a neighborhood $U$ of $p$ in $\cB$ and a $C^1$ coordinate system $(x_1,x_2,y,z)$ on $U$, centered at $p$, and such that, in this coordinate system, $X_{\cB}$ reads: 
\begin{equation}
\label{e.Takens-Bianchi}
X_{\cB}\left(x_1,x_2,y,z\right)=\lambda^s_1(y)x_1\frac{\partial}{\partial x_1}+\lambda^s_2(y)x_2\frac{\partial}{\partial x_2}+\lambda^u(y)z\frac{\partial}{\partial z}
\end{equation}
where $\lambda^s_1(y)<\lambda^s_2(y)<0$ and $\lambda^u(y)>0$ for every $y$.
\end{proposition}

\begin{remarks}
\label{r.linearisation-properties}
Let $U$ be a neighborhood of the point $p$ and $\left(x_1,x_2,y,z\right)$ be a $C^1$ coordinate system on $U$ centered at $p$, such that $X_\cB$ satisfies~\eqref{e.Takens-Bianchi} with $\lambda^s_1(y)<\lambda^s_2(y)<0$ and $\lambda^u(y)>0$. Then:
\begin{enumerate}
\item[(i)]  The vector field $X_{\cB}$ vanishes on the one-dimensional submanifold $\{x_1=x_2=z=0\}$ and nowhere else. It follows that this submanifold is the intersection of the Kasner circle with $U$~:
$$\cK\cap U =  \{x_1=x_2=z=0\}.$$
\item[(ii)] For every $y$, the real numbers $\lambda^s_1(y),\lambda^s_2(y),\lambda^u(y)$ are the three non-zero eigenvalues of the derivative of $X_{\cB}$ at the point of coordinates $(0,0,y,0)\in\cK$. Recall that this derivative also has one zero eigenvalue (corresponding to the direction of the Kasner circle). 
\item[(iii)] For every $\zeta$, the three-dimensional sub-manifold $\{y=\zeta\}$ is invariant under the flow of $X_{\cB}$ and the restriction of $X_{\cB}$ to this submanifold is linear. 
\item[(iv)] The three-dimensional submanifolds $\{x_1=0\}$, $\{x_2=0\}$ and $\{z=0\}$ are invariant under the flow of $X_{\cB}$, and contain $\cK\cap U$. It follows that these submanifolds coincide up to permutation with the submanifolds $\{N_1=0\}$, $\{N_2=0\}$ and $\{N_3=0\}$. 
As a consequence, the two-dimensional submanifolds $\{x_1=x_2=0\}$, $\{x_1=z=0\}$ and \hbox{$\{x_2=z=0\}$} coincide up to permutation with the submanifolds $\{N_1=N_2=0\}$, \hbox{$\{N_1=N_3=0\}$} and $\{N_2=N_3=0\}$. In particular, 
$$(\cB_{\2}\cap\cK)\cap U= \{x_1=x_2=0\}\cup \{x_1=z=0\}\cup \{x_2=z=0\}.$$
\item[(v)] \label{i.centralizer} The right-hand side of~\eqref{e.Takens-Bianchi} is unchanged if one replaces $x_1$ (resp. $x_2$ and $z$) by $-x_1$ (resp. by $-x_2$ and $-z$). Therefore, we may assume that 
$$\cB^+\cap U = \{ x_1\geq 0 \,,\, x_2\geq 0 \,,\, z\geq 0 \}.$$
\item[(vi)]  On the one hand, according to item~(i), the metric induced on the one-dimensional submanifold $\cK\cap U$ by the riemannian metric $(dx_1)^2+(dx_2)^2+(dy)^2+(dz)^2$ is simply $(dy)^2$. On the other hand, the right-hand side of~\eqref{e.Takens-Bianchi} is unchanged if one replaces $y$ by $\varphi(y)$ where $\varphi:\RR\to\RR$ is any  diffeomorphism such that $\varphi(0)=0$. Therefore, up to replacing the coordinate $y$ by $\varphi(y)$ for some appropriate diffeomorphism $\varphi$, we may assume that the metrics induced on the piece of circle $\cK\cap U$ by the riemannian metrics $(dx_1)^2+(dx_2)^2+(dy)^2+(dz)^2$ andby the riemannian metric  $h=(d\Sigma_1)^2+(d\Sigma_2)^2+(d\Sigma_3)^2+(dN_1)^2+(dN_2)^2+(dN_3)2$ coincide. 
\end{enumerate}
\end{remarks}

\begin{proof}[Proof of proposition~\ref{p.linearisation}]
The derivative $DX_{\cB}(p)$ has two negative, one zero, and one positive eigenvalue. According to proposition~\ref{p.characterisation-Sternberg}, the vector field $X_{\cB}$ satisfies Sternberg's condition at $p$. Therefore, a  crude application of Takens' theorem~\ref{t.Takens} implies that there exists a $C^1$ local coordinate system $(x_1,x_2,y,z)$ on a neighborhood $U$ of $p$ in $\cB$, centered at $p$, such that~:
\begin{equation}
\label{e.crude-Takens}
X_{\cB}(x_1,x_2,y,z)=\sum_{i=1}^2\sum_{j=1}2a_{i,j}(y)x^j\frac{\partial}{\partial x^i}+\phi(y)\frac{\partial}{\partial y}+b(y)z\frac{\partial}{\partial z}
\end{equation}
for some real valued functions $\phi,a_{1,1},a_{1,2},a_{2,1},a_{2,2},b$ defined on a neighborhood of $0$ in $\RR$. Moreover, the eigenvalues of the matrix $(a_{i,j}(0))$ are negative, and $b(0)$ is positive. 
Replacing $U$ by a smaller neighborhood of $p$ if necessary, we can assume that the three Taub points $T_1,T_2,T_3$ are not in $U$.

Now, \eqref{e.crude-Takens} implies that the curve $\{x_1=x_2=z=0\}$ is the only curve in $U$ containing the point $p$, invariant under the flow of $X_{\cB}$, and such that $DX_{\cB}(p)$ vanishes on the tangent space at $p$ of this curve. Hence, the curve $\{x_1=x_2=z=0\}$ has to be the intersection of the Kasner circle $\cK$ with $U$. Since $X_{\cB}$ vanishes on $\cK$, it follows that $\phi=0$. 

For $\zeta$ small enough, let $q_\zeta$ be the point of coordinates $(x_1,x_2,y,z)=(0,0,\zeta,0)$. This is a point of the Kasner circle $\cK$, which is not a Taub point. Hence, the derivative $DX(q_\zeta)$ has four distinct eigenvalues~: two distinctive negative eigenvalues $\lambda^s_1(\zeta)<\lambda^s_2(\zeta)<0$, one zero eigenvalue associated  to the direction of the Kasner circle, and one positive eigenvalue $\lambda^u(\zeta)>0$. The set $(y=\zeta)$ is a three-dimensional manifold, transversal to the Kasner circle. Looking at~\eqref{e.crude-Takens}, we see that this three-dimensional submanifold is invariant under the flow of $X_{\cB}$, and that the restriction of $X_{\cB}$ to this invariant manifold is linear.This shows that $b(\zeta)=\lambda^u(\zeta)$, and that $\lambda^s_1(\zeta)$ and $\lambda^s_2(\zeta)$ are the eigenvalues of the matrix $(a_{i,j}(\zeta))_{i,j=1,2}$. Since $\lambda^s_1(\zeta)$ and $\lambda^s_2(\zeta)$ are distinct, there exists a linear change of coordinates $(x_1,x_2,z)\rightarrow (\hat x_1,\hat x_2,z)$ on the submanifold $(y=\zeta)$, so that 
$$X_{\cB}(\hat x_1,\hat x_2,\zeta,z)=\lambda^s_1(\zeta)\hat x_1\frac{\partial}{\partial \hat x_1}+\lambda^s_2(\zeta)\hat x_2\frac{\partial}{\partial \hat x_2}+\lambda^u(\zeta)z\frac{\partial}{\partial z}.$$ 
Since eigenvalues and eigendirections of the point $q_{\zeta}$ depend in a smooth way of $\zeta$, one may perform the above change of coordinates simultaneously for every $\zeta$, and get a $C^1$ coordinate system $(\hat x_1,\hat x_2,y,z)$ defined on $U$, such that
$$X_{\cB}(\hat x_1,\hat x_2,y,z)=\lambda^s_1(y)\hat x_1\frac{\partial}{\partial \hat x_1}+\lambda^s_2(y)\hat x_2\frac{\partial}{\partial \hat x_2}+\lambda^u(y)z\frac{\partial}{\partial z}.$$ 
The proposition is proven.
\end{proof}

\subsection{Choice of a linearizing coordinate system}
\label{ss.choice-neighborhood}
From now on, for every point $p$ of the Kasner circle which is not preperiodic under the Kasner map, we fix a neighborhood $U^p$ of $p$ in $\cB$, and a $C^1$ local coordinate system $(x^p_1,x^p_2,y^p,z^p)$ on $U^p$, centered at $p$, such that, in this coordinate system, the Wainwright-Hsu vector field $X_{\cB}$ reads: 
\begin{equation}
\label{e.def-local-coordinate-system}
X_{\cB}\left(x_1^p,x_2^p,y^p,z^p\right)=\lambda^s_1(y^p)x_1^p\frac{\partial}{\partial x_1^p}+\lambda^s_2(y^p)x_2^p\frac{\partial}{\partial x_2^p}+\lambda^u(y^p)z^p\frac{\partial}{\partial z^p}
\end{equation} 
with $\lambda^s_1(y^p)<\lambda^s_2(y^p)<0$ and $\lambda^u(y^p)>0$. According to item~(ii) of remarks~\ref{r.linearisation-properties}, the real numbers $\lambda^s_1(y^p)$, $\lambda^s_2(y^p)$, $\lambda^u(y^p)$ are the non-zero eigenvalues of the point $(0,0,y^p,0)$. According to items (i), (iv), (v) of remarks~\ref{r.linearisation-properties}, up to changing the sign of the coordinates $x^p_1$, $x^p_2$ and $z^p$, we may (and we will) assume that
\begin{eqnarray}
\label{e.K}
\cK\cap U^p & = & \{x^p_1=x^p_2=z=0\}\\
\label{e.B_II}
(\cB_{\2}\cup\cK)\cap U^p & = & \{x^p_1=x^p_2=0\}\cup \{x_1^p=z^p=0\}\cup \{x_2^p=z^p=0\},\\
\label{e.B+}
\cB^+\cap U^p &= &\{ x^p_1\geq 0 \,,\, x^p_2\geq 0 \,,\, z^p\geq 0 \}
\end{eqnarray}
Note that the derivative $DX_{\cB}(q)$ has three non-zero eigenvalues for every point $q\in \cK\cap U^p$; this shows that the neighborhood $U^p$ is disjoint from the three special points. We will consider the riemannian metric $g^p$ on $U^p$ defined by
\begin{equation}
\label{e.g^p}
g^p:=\left(dx^p_1\right)^2+\left(dx^p_2\right)^2+\left(dy^p\right)^2+\left(dz^p\right)2.
\end{equation}
According to item (vi) of remarks~\ref{r.linearisation-properties}, up to replacing the coordinate $y^p$ by $\varphi(y^p)$ for some appropriate diffeomorphism $\varphi$, we may (and we will) assume that $g^p$ induces the same metric on the piece of Kasner circle $\cK\cap U^p$ as the riemannian metric $h=(d\Sigma_1)^2+(d\Sigma_2)^2+(d\Sigma_3)^2+(dN_1)^2+(dN_2)^2+(dN_3)^2$.


\section[Dulac map for Wainwright-Hsu vector field]{Dulac map for Wainwright-Hsu vector field near a point of the Kasner circle which is not preperiodic under the Kasner map}
\label{s.Dulac}

Let $p$ be a point of the Kasner circle which is not preperiodic for the Kasner map. The purpose of the present section is to analyze the behavior of the orbits of the Wainwright-Hsu vector field $X_\cB$ close to $p$. More precisely, we want to consider an orbit of $X_{\cB}$ which passes close to $p$, and to study the evolution of the distance from this orbit to the mixmaster attractor $\cA=\cK\cup\cB_{\2}$, as well as the drift of this orbit in the direction tangent to the mixmaster attractor.  

\subsection{The flow of $X_{\cB}$ inside $U^p$}
We consider the neighborhood $U^p$, and the coordinate system $(x^p_1,x^p_2,y^p,z^p)$ defined in subsection~\ref{ss.choice-neighborhood}. Using the expression~\eqref{e.def-local-coordinate-system}, one can calculate explicitly the time~$t$map $X_{\cB}^t$ of the flow of the Wainwright-Hsu vector field $X_{\cB}$ in the $(x^p_1,x^p_2,y^p,z^p)$ coordinate system. It reads
\begin{equation}
\label{e.linear-flow}
X_{\cB}^t\left(x^p_1,x^p_2,y^p,z^p\right)= \left(x^p_1 e^{\lambda^s_1(y^p)t} \,,\, x^p_2 e^{\lambda^s_2(y^p)t} \,,\, y^p \,,\, z^p e^{\lambda^u(y^p)t} \right).
\end{equation}
Of course, this expression is only valid as long as the orbit remains in the neighborhood $U^p$.

\subsection{The box $V^p$}
Now, we fix some constants $\alpha<\beta$ and $\gamma>0$, and we consider the subset $V^p$ of~$U^p\cap\cB^+$ defined by:
\begin{equation}
\label{e.V^p}
V^p=V^p(\alpha,\beta,\gamma)=\{ 0\leq x^p_1\leq \gamma \;,\; 0\leq x^p_2\leq \gamma \;,\; 0\leq z^p\leq \gamma \;\mbox{ and }\;\alpha\leq y^p\leq\beta\}.
\end{equation}
We assume that $\alpha,\beta,\gamma$ are small enough, so that $V^p$ is contained in the interior of $U^p$. 

\begin{remark}
\label{r.CNS-neighborhood}
The set $V^p$ is a neighborhood of the point $p$ in $\cB^+$ if and only if $\alpha<0<\beta$. It is important to note that the results of the present section are valid even if $p$ is not in $V^p$. 
\end{remark}

In the $(x^p_1,x^p_2,y^p,z^p)$ coordinate system, the set $V^p$ is a 4-dimensional box, \emph{i.e.} the cartesian product of four closed intervals. The boundary of $V^p$ is made of eight faces. Three of these eight faces will play an important role in the remainder of the paper: 
\begin{equation}
\label{e.MM'N}
M^p_1:=V^p\cap \{x^p_1=\gamma\}\quad\quad\quad M^p_2  :=  V^p\cap \{x^p_2=\gamma\}\quad\quad\quad N^p :=  V^p\cap \{z^p=\gamma\}.
\end{equation}
Looking at~\eqref{e.def-local-coordinate-system}, we notice that $X_{\cB}$ is transversal to $M^p_1$, $M^p_2$ and $N^p$, and is tangent to the five other faces of~$V^p$. Moreover, we notice that $X_{\cB}$ is pointing inward $V^p$ along~$M^p_1$ and~$M^p_2$; it is pointing outward~$V^p$ along~$N^p$. It follows that:
\begin{itemize}
\item an orbit of $X_{\cB}$ can enter in $V^p$ by crossing either the face $M^p_1$ or by crossing the face $M^p_2$;
\item an orbit of $X_{\cB}$ can only exit $V^p$ by crossing the face $N^p$.
\end{itemize}

\subsection{Behavior of type $\2$ orbits}
We will study the behavior of the orbits of $X_{\cB}$ in $V^p$. First, we focus our attention on type $\2$ orbits. We want to understand which type $\2$ orbit intersect the hypersurfaces (with boundary and corners) $M^p_1$, $M^p_2$ and $N^p$. Recall that every type $\2$ orbit $\cO$ of $X_\cB$ is a heteroclinic orbit connecting a point $\alpha(\cO)\in\cK\setminus\{T_2,T_2,T_3\}$ to the point $\omega(\cO)=f(\alpha(\cO))\in\cK\setminus\{T_2,T_2,T_3\}$.

\begin{proposition}
\label{p.behavior-II}
Let $q$ be a point $\cB_{\2}^+$, and denote by $\cO$ the orbit of $q$. Let $\alpha(q)=\alpha(\cO)$ denote the unique $\alpha$-limit point of $\cO$, and $\omega(q)=\omega(\cO)=f(\alpha(\cO))$ denote the unique $\omega$-limit point of $\cO$. 
\begin{enumerate}
\item The orbit $\cO$ intersects the hypersurface $N^p$ if and only if the point $\alpha(q)$ is in $V^p$. 
\item  The orbit of $\cO$ intersects the hypersurface $M^p_1\cup M^p_2$ if and only if the point $\omega(q)$ is in $V^p$.
\end{enumerate}
\end{proposition}

\begin{proof}
We prove the first statement; the second one follows from similar arguments. Of course, we will work in the $(x^p_1,x^p_2,y^p,z^p)$ coordinate system. According to~\eqref{e.K},~\eqref{e.B_II},~\eqref{e.V^p} and~\eqref{e.MM'N},
\begin{eqnarray*}
\cK\cap V^p & = & \{ x^p_1=x^p_2=0 \,,\, \alpha\leq y^p\leq\beta \,,\,z^p=0\},\\
\cB_{\2}\cap N^p& = & \{ x^p_1=x^p_2=0 \,,\, \alpha\leq y^p\leq\beta \,,\,z^p=\gamma\}.
\end{eqnarray*}
Suppose that the orbit $\cO$ intersects the hypersurface $N^p$ at some point  $\bar q=(0,0,\zeta,\gamma)$, with $\alpha\leq \zeta\leq\beta$. Then, according to~\eqref{e.linear-flow}, the past orbit of $\bar q$ is contained in $V^p$, and converges to the point  $\alpha(q)=(0,0,\zeta,0)\in\cK\cap V^p$. In particular, the point $\alpha(q)$ is in~$V^p$. Conversely, suppose that the point $\alpha(q)$ is in $V^p$. Then  $\alpha(q)=(0,0,\zeta,0)$ for some $\zeta\in [\alpha,\beta]$. Using again~\eqref{e.linear-flow}, we see that the only orbit of $X_{\cB}$ in $V^p$ converging towards the point $(0,0,\zeta,0)$ as $t\to-\infty$ is the curve $t\mapsto (0,0,\zeta,e^{\lambda^u(\zeta)t})$. Hence, the orbit $\cO$ intersects the hypersurface $N^p$ at the point $\bar q=(0,0,\zeta,\gamma)$. 
\end{proof}

This proposition allows to define two maps 
$$\begin{array}[t]{crclccrclc}
\alpha^p : &  \cB_{\2}^+\cap N^p & \longrightarrow  & \cK\cap V^p & \quad\quad & \omega^p :  & \cB_{\2}^+\cap (M^p_1\cup M^p_2) & \longrightarrow & \cK\cap V^p\\
& q & \longmapsto & \alpha(q) & & & q & \longmapsto & \omega(q)
\end{array}$$
The map $\alpha^p$ is one-to-one (there is only one type $\2$ orbit in $\cB^+$ which ``starts" at a given point of $\cK\setminus\{T_2,T_2,T_3\}$), whereas the map $\omega^p$ is two-to-one (there are two type $\2$ orbits in $\cB^+$ which ``arrive" at a given point of~$\cK\setminus\{T_2,T_2,T_3\}$). The restriction of $\omega^p$ to $\cB_{\2}^+\cap M^p_1$ (resp. to $\cB_{\2}^+\cap M^p_2$) is one-to-one.

\begin{proposition}
\label{p.isometries}~
The maps $\alpha^p$ and $\omega^p$ are local isometries with respect to the metrics induced on $\cB_{\2}^+\cap N^p$, $ \cB_{\2}^+\cap (M^p_1\cup M^p_2)$ and $\cK\cap V^p$ by the riemannian metric $g^p=(dx^p_1)^2+(dx^p_2)^2+(dy^p)^2+(dz^p)2$. 
\end{proposition}

\begin{proof}
The proof of proposition~\ref{p.behavior-II} shows that, in the $(x^p_1,x^p_2,y^p,z^p)$ coordinate system, the map $\alpha^p$ reads $\alpha^p(0,0,y^p,\gamma)=(0,0,y^p,0)$. Similarly, the map $\omega^p$ reads $\omega^p(\gamma,0,y^p,0)= (0,0,y^p,0)$. 
\end{proof}

\subsection{The Dulac map $\Delta^p_1:M^p_1\to N^p$}
Now, we want to study the behavior of arbitrary orbits of $X_{\cB}$ which enter in $V^p$. Let $q$ be a point on the face $M^p_1$. Denote by $(\gamma,x_2^p,y^p,z^p)$ the coordinates of $q$. If the $z^p>0$ (which is typically the case if $q\in\cB_{\9}$), then~\eqref{e.linear-flow} shows that the forward orbit of $q$ will eventually exit $V^p$ by crossing the face $N^p$. If $z^p=0$ (which is typically the case if $q\in\cB_{\2}$), then~\eqref{e.linear-flow} shows that the forward orbit of $q$ will remain in $V^p$ forever. It will converge towards the point $\omega(q)\in\cK\cap V^p$. According to proposition~\ref{p.behavior-II}, the heteroclinic orbit $\cO_{\omega(q),f(\omega(q))}$ will eventually exit~$V^p$, by crossing the face $N^p$. So, we may define a map $\Delta^p_1\;:\;M^p_1\longrightarrow N^p$ as follows~:
\begin{itemize}
\item if $z^p>0$ then $\Delta^p_1(q)$ is the first intersection point of the orbit of $q$ with the hypersurface$N^p$~;
\item if $z^p=0$ then $\Delta^p_1(q)$ is the first (and unique) intersection point of the type $\2$ heteroclinic orbit $\cO_{\omega(q), f(\omega(q))}$ with  the hypersurface $N^p$.
\end{itemize}
We call $\Delta^p_1$  a \emph{Dulac map} since it is the exact analog, in our situation, of the classical Dulac maps used to study planar vector fields. 
Formula~\eqref{e.linear-flow} show that, in the $(x^p_1,x^p_2,y^p,z^p)$ coordinate system, the map $\Delta^p_1$ reads~:
\begin{eqnarray}
\label{e.Dulac-1}
\Delta^p_1\left (\gamma \,,\, x^p_2 \,,\, y^p \,,\, z^p \right) & = & \left( \gamma.\left(\frac{z^p}{\gamma}\right)^{-\frac{\lambda^s_1(y^p)}{\lambda^u(y^p)}} \;,\; x^p_2.\left(\frac{z^p}{\gamma}\right)^{-\frac{\lambda^s_2(y^p)}{\lambda^u(y^p)}} \;,\; y^p \;,\; \gamma\right) \mbox{ if } z^p>0 \\
\label{e.Dulac-2}
\Delta^p\left (\gamma \,,\, x^p_2 \,,\, y^p \,,\, 0 \right) & = & \left( 0 \,,\, 0 \,,\, y^p \,,\, \gamma\right)
\end{eqnarray}

\begin{remark}
Given a point $q$ in $M^p_1$ such that $z^p(q)>0$, one can consider the \emph{exit time} of $q$, \emph{that is} the real number $t(q)$ such that $\Delta^p_1(q)=X^{t(q)}(q)$. Using~\eqref{e.linear-flow}, it is easy to see 
$$t\left (\gamma \,,\, x^p_2 \,,\, y^p \,,\, z^p \right) = -\frac{1}{\lambda^u(y^p)}\log\left(\frac{z^p}{\gamma}\right)$$
\end{remark}

For every $q\in M^p_1$, we decompose $T_q M^p$ as a direct sum of two linear subspaces $F^s_q \oplus F^u_q$ where 
\begin{eqnarray}
F^s_q & := & \RR.\frac{\partial}{\partial x^p_2}(q)\oplus\RR.\frac{\partial}{\partial z^p}(q)\\
F^u_q & := & \RR. \frac{\partial}{\partial y^p}(q).
\end{eqnarray}
Similarly, for every $r\in N^p$, we decompose $T_r N^p$ as a direct sum of two linear subspaces $G^s_r \oplus F^u_r$ where 
\begin{eqnarray}
G^s_r & := & \RR.\frac{\partial}{\partial x^p_1}(r)\oplus\RR.\frac{\partial}{\partial x^p_2}(r)\\
G^u_r & := & \RR. \frac{\partial}{\partial y^p}(r).
\end{eqnarray}
We can now state the properties of the Dulac map $\Delta^p_1$ which will be the core of our proof of theorem~\ref{t.main}:

\begin{proposition}
\label{p.hyperbolicity-Phi}
The Dulac map $\Delta^p_1:M^p_1\rightarrow N^p$ is $C^1$. Moreover, for every point $q\in M^p_1\cap \{z^p=0\}$, the derivative of the map $\Delta^p_1$ at $q$ satisfies: 
\begin{itemize}
\item $D\Delta^p_1(q).v=0$ for every vector $v\in F^s_q$;\smallskip
\item $D\Delta^p_1(q)$ maps $F^u_q$ on $G^u_{\Delta^p_1(q)}$, and $\|D\Delta^p_1(q).v\|_{g^p}=\|v\|_{g^p}$ for every vector $v\in F^u_q$.
\end{itemize}
\end{proposition}

This proposition roughly says the following: when an orbit of $X_{\cB}$ passes close to the point $p\in\cK$, the distance from this orbit to the mixmaster attractor $\cA=\cK\cup\cB_{\2}$ is contracted super-linearly (this distance is measured by the  coordinates $x^p_1$, $x^p_2$ and $z^p$), whereas there is no drift in the direction tangent to the attractor (this drift is measured by the coordinate $y^p$). The key point of the proof of proposition~\ref{p.hyperbolicity-Phi} is the following elementary observation:

\begin{lemma}
\label{l.domination}
For every $y^p$, we have $\displaystyle \left|\frac{\lambda^s_1(y^p)}{\lambda^u(y^p)}\right|>1$ and $\displaystyle \left|\frac{\lambda^s_2(y^p)}{\lambda^u(y^p)}\right|>1$. 
\end{lemma}

This lemma says that, at every point $\bar q$ of $\cK\setminus\{T_2,T_2,T_3\}$, the positive eigenvalue of the derivative $DX(\bar q)$ is dominated by the contracting eigenvalues.

\begin{proof}[Proof of lemma~\ref{l.domination}]
Fix $y^p$, and denote by $q$ be the point of coordinates $(0,0,y^p,0)$ in the $(x^p_1,x^p_2,y^p,z^p)$ coordinate system. Denote by $(\Sigma_1,\Sigma_2,\Sigma_3,0,0,0)$ the  coordinates of $q$ in the Wainwright-Hsu coordinate system, and by $u$ the Kasner parameter of $q$. The real numbers $\lambda^s_1(y^p)$, $\lambda^s_2(y^p)$, $\lambda^u(y^p)$ are the three non-zero eigenvalues the derivative $DX_{\cB}(q)$. Hence, these numbers are equal up to permutation to $-(2+\Sigma_1)$, $-(2+\Sigma_2)$, $-(2+\Sigma_3)$. Using~\eqref{e.eigenvalues-Kasner-parameter} and the inequalities $\lambda^s_1(y^p)<\lambda^s_2(y^p)<0<\lambda^u(y^p)$, we deduce that
$$(\lambda^s_1(y^p)\,,\,\lambda^s_2(y^p)\,,\,\lambda^u(y^p))=\left(\frac{-6u(1+u)}{1+u+u^2}\,,\,\frac{-6(1+u)}{1+u+u^2}\,,\,\frac{6u}{1+u+u^2}\right).$$
The lemma follows since $\frac{1+u}{u}>1$ and $(1+u)>1$ for every $u\geq 1$.
\end{proof}

\begin{proof}[Proof of proposition~\ref{p.hyperbolicity-Phi}]
The fact that $\Delta^p_1$ is $C^1$ (and even analytical) in restriction to $M^p_1\cap \{z^p>0\}$ is an immediate consequence of formula~\eqref{e.Dulac-1}. The fact that $\Delta^p_1$ is $C^1$ on $M^p_1\cap \{z^p=0\}$ follows from~\eqref{e.Dulac-1},~\eqref{e.Dulac-2} and lemma~\ref{l.domination}. The same ingredients actually show that, for every point $q\in M^p\cap \{z^p=0\}$,
$$ D\Delta^p_1(q).\frac{\partial}{\partial x^p_2}(q) = D\Delta^p_1(q).\frac{\partial}{\partial z^p}(q) = 0\quad\mbox{ and } D\Delta^p_1(q).\frac{\partial}{\partial y^p}(q) = \frac{\partial}{\partial y^p}(\Delta^p_1(q)).$$
The proposition follows.
\end{proof}

\begin{remark}
We do not know if $\Delta^p_1$ is $C^{1+\epsilon}$ for any given $\epsilon>0$, unless we have some \emph{a priori}  lower bounds for the distance between the ratios $\frac{\lambda^s_1(y^p)}{\lambda^u(y^p)}$ and $\frac{\lambda^s_2(y^p)}{\lambda^u(y^p)}$ and $1$. This is the reason why, in the statement of theorem~\ref{t.main}, we cannot guarantee that the stable manifold $W^s(q)$ contains a $C^{1+\epsilon}$-embedded disc for any given $\epsilon$. Actually such an $\epsilon$ exists for every $q$ but it does depend on $q$, and tends to $0$ if $q$ approaches one of the three special points tends to $0$.
\end{remark}

\subsection{The Dulac map $\Delta^p_2:M^p_2\to N^p$}
The coordinates $x^p_1$ and $x^p_2$ play similar roles in the expression of the vector field $X_\cB$ and in the definition of $V^p$.  So, we have to consider a second Dulac map~$\Delta^p_2:M^p_2\longrightarrow N^p$ defined as follows. Let $q$ be a point on the face $M^p_1$. Denote by $(\gamma,x_2^p,y^p,z^p)$ the coordinates of $q$. 
\begin{itemize}
\item if $z^p(q)>0$ then $\Delta^p_2(q)$ is the first intersection point of the orbit of $q$ with the hypersurface$N^p$~;
\item if $z^p(q)=0$ then $\Delta^p_2(q)$ is the first (and unique) intersection point of the type $\2$ heteroclinic orbit $\cO_{\omega(q), f(\omega(q))}$ with  the hypersurface $N^p$.
\end{itemize}
In the $(x^p_1,x^p_2,y^p,z^p)$ coordinate system, the Dulac map $\Delta^p_2$ reads:
\begin{eqnarray}
\label{e.Dulac-1'}
\Delta^p_2\left ( x^p_1 \,,\, \gamma \,,\, y^p \,,\, z^p \right) & = & \left( x^p_1.\left(\frac{z^p}{\gamma}\right)^{-\frac{\lambda^s_1(y^p)}{\lambda^u(y^p)}} \;,\; \gamma.\left(\frac{z^p}{\gamma}\right)^{-\frac{\lambda^s_2(y^p)}{\lambda^u(y^p)}}\;,\; y^p \;,\; \gamma\right) \mbox{ if } z^p>0 \\
\label{e.Dulac-2'}
\Delta^p_2\left ( x^p_1 \,,\, \gamma \,,\, y^p \,,\, 0 \right) & = & \left( 0 \,,\, 0 \,,\, y^p \,,\, \gamma\right)
\end{eqnarray}
For every $q\in M^p_2$, we will write $T_q M^p_2$ as a direct sum of two linear subspaces $F^s_q \oplus F^u_q$ where 
\begin{eqnarray}
F^s_q & := & \RR.\frac{\partial}{\partial x^p_1}(q)\oplus\RR.\frac{\partial}{\partial z^p}(q)\\
F^u_q & := & \RR. \frac{\partial}{\partial y^p}(q).
\end{eqnarray}
Then, we can summarize the key properties of the map $\Delta^p_2$ as follow~: 

\begin{proposition}
\label{p.hyperbolicity-Phi'}
The map $\Delta^p_2:M^p_2\to N^p$ is $C^1$. Moreover, for every $q\in M^p_2\cap \{z^p=0\}$,  the derivative of of the map $\Delta^p_2$ at $q$ satisfies~:
\begin{itemize}
\item $D\Delta^p_2(q).v=0$ for every vector $v\in F^s_q$~;\smallskip
\item $D\Delta^p_2(q)$ maps $F^u_q$ on $G^u_{\Delta^p_2(q)}$, and $\|D\Delta^p_2(q).v\|_{g^p}=\|v\|_{g^p}$ for every vector $v$ in $F^u_q$.
\end{itemize}
\end{proposition}


\section[Construction of a ``Poincar\'e map"]{Construction of a ``Poincar\'e map" associated to a closed forward-invariant aperiodic set of the Kasner circle} 
\label{s.construction-Poincare-section}

\emph{From now on until the end of section~\ref{s.end-proof-main}, we consider a closed  forward-invariant aperiodic subset $C$ of the Kasner circle $\cK$.} 

Observe that $C$ is necessarily totally discontinuous. Indeed the points of $\cK$ that are preperiodic under the Kasner map $f$ are dense in~$\cK$ (this can be proved in several different ways; this follows for example the equivalence $2\Leftrightarrow 3$ of proposition~\ref{p.characterisation-Sternberg}). We shall denote by $\wC$ the union of $C$ and all the type $\2$ orbits connecting two points of $C$:
$$\wC:=C\cup \bigcup_{q\in C} \cO_{q,f(q)}.$$
We want to prove that, for every point~$q\in C$, the stable manifold $W^s(q)$ contains a 3-dimensional disc (see definition~\ref{d.stable-manifold} and theorem~\ref{t.main}). To this end, we will consider a kind of ``Poincar\'e section" for $\wC$, and study the Poincar\'e map associated with this section\footnote{The set $\wC$ cannot admit a true Poincar\'e section, since it contains some singularities of $X_{\cB}$ (namely, the points of $C$). Nevertheless, we will consider a hypersurface $N$ such that every type $\2$ orbit in $\wC$ intersects $N$ transversally. The hypersurface $N$ will play the role of a Poincar\'e section.}.

For every $p\in C$, we consider a neighborhood $U^p$ of $p$, and a local coordinate system $(x^p_1,x^p_2,y^p,z^p)$ on $U^p$, as in the previous section. For each $p\in C$, we choose $\alpha^p_0<0<\beta^p_0$ and $\gamma^p_0>0$ small enough, so that 
$$V^p_0:=V^p(\alpha^p_0,\beta^p_0,\gamma^p_0)=\{ 0\leq x^p_1\leq \gamma^p_0 \;,\; 0\leq x^p_2\leq \gamma^p_0 \;,\; 0\leq z^p\leq \gamma^p_0 \;\mbox{ and }\;\alpha^p_0\leq y^p\leq\beta^p_0\}.$$
is contained in the interior of $U^p$. Up to slightly modifying  $\alpha^p_0,\beta^p_0,\gamma^p_0$ we can assume that 
 the boundary of $V^p_0$ is disjoint from $C$ (\emph{i.e} that the points of coordinates $(0,0,\alpha,0)$ and $(0,0,\beta,0)$ in the $(x^p_1,x^p_2,y^p,z^p)$ coordinate system are not  in $C$): this is possible since $C$ is totally discontinuous. Observe that $V^p_0$ is a neighborhood of $p$, since $\alpha^p_0<0<\beta^p_0$ (see remark~\ref{r.CNS-neighborhood}). 

Since $C$ is compact, one can find a finite number of points $p_1,\dots,p_n\in C$ such that the neighborhoods $V^{p_1}_0,\dots,V^{p_n}_0$ cover $C$. Now, we modify these neighborhoods in order to make them pairwise disjoint:
\begin{itemize}
\item we set $(\alpha^{p_1},\beta^{p_1},\gamma^{p_1}):=(\alpha^{p_1}_0,\beta^{p_1}_0,\gamma^{p_1}_0)$, and $V^{p_1}:=V^{p_1}(\alpha^{p_1},\beta^{p_1},\gamma^{p_1})= V^{p_1}_0$;
\item then, we can find some constants $\alpha_2,\beta_2,\gamma_2$ such that $V^{p_2}:=V^{p_2}(\alpha_2,\beta_2,\gamma_2)$ is contained in $V^{p_2}_0\setminus V^{p_1}$, and such that $C\cap (V^{p_2}_0\setminus V^{p_1})$ is contained in the interior of $V^{p_2}$;
\item then, we can find some constants $\alpha_3,\beta_3,\gamma_3$ such that $V^{p_3}:=V^{p_3}(\alpha_3,\beta_3,\gamma_3)$ is contained in $V^{p_3}_0\setminus (V^{p_1}\cup V^{p_2})$, and such that $C\cap (V^{p_3}_0\setminus (V^{p_1}\cup V^{p_2}))$ is contained in the interior of $V^{p_3}$;
\item etc.
\end{itemize}
At the end of this process, we get $n$ pairwise disjoint domains $V^{p_1},\dots,V^{p_n}$, such that $C$ is contained in the interior of $V^{p_1}\cup\dots\cup V^{p_n}$. For each $i$, $V^{p_i}$ is contained in the interior of $U^{p_i}$, and there are some constants $\alpha_i,\beta_i,\gamma_i$ such that $V^{p_i}=V^{p_i}(\alpha_i,\beta_i,\gamma_i)$. Hence, the result of section~\ref{s.Dulac} apply to $V^{p_i}$. It may happen that, for some $i$, the point $p_i$ is not in $V^{p_i}$ (\emph{i.e.} that $\alpha_i$ or $\beta_i$ is non-positive), but we do not care.  

Now, we denote
\begin{eqnarray*}
V^C & := & V^{p_1}\sqcup\dots\sqcup V^{p_n},\\
M_1^C & :=  & M^{p_1}_1\sqcup\dots\sqcup M^{p_n}_1,\\
M_2^C & := & M^{p_1}_2\sqcup\dots\sqcup M^{p_n}_2,\\
M^C & := & M_1^C\cup M_2^C\\
N^C & := & N^{p_1}\sqcup\dots\sqcup N^{p_n}.
\end{eqnarray*}
 Then $V^C$ is a neighborhood of $C$ in $\cB^+$. The hypersurfaces $M_1^C$, $M_2^C$ and $N^C$ are transverse to $X_{\cB}$. An orbit of $X_{\cB}$ can only enter in $V^C$ by crossing $M^C=M_1^C\cup M_2^C$, and can only exit $V^C$ by crossing $N^C$. Moreover, according to proposition~\ref{p.behavior-II}, we have the following important properties~: 

\begin{proposition} 
\label{p.orbits-of-C}
Every type $\2$ orbit whose $\omega$-limit point is in $C$ intersects $M^C$. Every type $\2$ orbit whose $\alpha$-limit point is in $C$ intersects $N^C$.
\end{proposition}

We will see $M^C$ as a kind of ``Poincar\'e section" for $\wC$. Let us define the ``Poincar\'e map" $\Phi$ associated to this section. First, we consider the ``Dulac map" 
$$\Delta^C:M^C\to N^C$$
defined by $\Delta^C_{|M^{p_i}_1}=\Delta^{p_i}_1$ and $\Delta^C_{|M^{p_i}_2}=\Delta^{p_i}_2$. For $q\in M^C$:
\begin{itemize}
\item if the forward orbit of $q$ exits $V^C$ by crossing $N^C$ (which is typically  the case if $q\in\cB_{\9}$), then $\Delta^C(q)$ is by definition the first intersection point of the orbit of $q$ with the hypersurface $N^C$~;
\item if the forward orbit of $q$ remains in $V^C$ forever (which is typically the case for every $q\in\cB_{\2}$), then $\Delta^C(q)$ is the first  intersection point of the type $\2$ heteroclinic orbit $\cO_{\omega(q),f(\omega(q))}$ with the hypersurface $N^C$.
\end{itemize}
Now, we consider the ``transition map" 
$$\Theta^C:N^C\to M^C$$
partially defined as follows. Given a point $q$ in $N^C$, if the forward orbit of $q$ re-enters in $V^C$, then $\Theta^C(q)$ is the first point of this forward orbit of $q$ which is in $V^C$ (this point is automatically on the hypersurface $M^C$); otherwise $\Theta^C(q)$ is not defined. The ``Poincar\'e map" $\Phi^C$ associated with the section $M^C$ is by definition the product of the ``Dulac map" $\Delta^C$ and the ``transition map" $\Theta^C$~:
$$\Phi^C:=\Theta^C\circ\Delta^C :M^C \to M^C.$$

In the next section, we will study the dynamics of the Poincar\'e map $\Theta^C$. For this purpose, we will use a riemannian metric on $g^C$ on $\cB$ such that, for $i=1,\dots,n$, 
$$g^C_{|V_i}=g^{p_i}=(dx^{p_i}_1)^2+(dx^{p_i}_2)^2+(dy^{p_i})^2+(dz^{p_i})^2.$$


\section[Stable manifolds for the Poincar\'e map]{Stable manifolds for the Poincar\'e map associated to a closed forward-invariant aperiodic set of the Kasner circle} 
\label{s.stable-manifolds-Poincare-map}

The purpose of this section is to prove that, for every point $q\in\widehat C\cap M^C$, the stable manifold of~$q$ for the ``Poincar\'e map" $\Phi^C$ contains a two-dimensional disc. To this end, we will prove that $\widehat C\cap M^C$ is a hyperbolic set for the map $\Phi^C$, and we will use a classical result on stable manifolds for hyperbolic sets. 

\begin{definition}
Let $(M,g)$ be a riemannian manifold and $\Phi:M\to M$ be a $C^1$ map. A \emph{hyperbolic set} for the map $\Phi$ is a compact $\Phi$-invariant subset $C$ of $M$ such that, for every $q\in C$, there is splitting $T_qM=F^s_q\oplus F^u_q$ which depends continuously on $q$ and such that, for some constant $\mu\in (0,1)$ and $nu>1$~:
\begin{eqnarray}
\label{e.contraction}
&& \mbox{$D\Phi(q).F^s_q\subset F^s_{\Phi(q)}$ and $\| D\Phi(x).v\|\leq \mu\| v\|$ for every $q\in C$ and $v\in F^s_q$}\\
\label{e.dilatation}
&& \mbox{ $D\Phi(q).F^u_q=F^u_{\Phi(q)}$ and $\|D\Phi(q).v\|\geq \nu\| v\|$ for every $q\in C$ and $v\in F^u_q$}.
\end{eqnarray}
The dimension of the vector space $F^s_q$ is called the \emph{index} of $C$. The constant $\mu$ is called a \emph{contraction rate} of $\Phi$ on $C$.
\end{definition}

\begin{theorem}(see e.g.~\cite[page 167]{PalisTakens1993})
\label{t.stable-manifold}
Let $\Phi:M\to M$ be a $C^1$ map of a manifold $M$, and $C$ be a compact subset of $M$ which is a hyperbolic of index $s$ for the map $\Phi$. Then, for every $\epsilon$ small enough, for every $q\in C$, the set 
$$W^s_\epsilon(\Phi,q):=\{r\in M \mid \mbox{dist}(\Phi^n(r),\Phi^n(r))\leq \epsilon \mbox{ for every }n\geq 0\}$$
is a $C^1$ embedded $s$-dimensional  disc, tangent to $F^s_q$ at $q$, depending continuously on $q$ (for the $C^1$ topology on the space of embeddings). Moreover, if $\mu$ is a contraction constant for $\Phi$ on $C$, then there exists a constant $\kappa$ such that, for every $\epsilon$ small enough, for every $q\in C$ and every $r\in W^s_\epsilon(\Phi,q)$, 
$$\mbox{dist}_g\left(\Phi^n(r),\Phi^n(q)\right)\leq \kappa \mu^n.$$
\end{theorem}

We want to apply this theorem to the Poincar\'e map $\Phi^C:M^C\to M^C$. So we need to prove that $\wC\cap M^C$
is a hyperbolic set for $\Phi^C$. Recall that $M^C= M_1^C\cup M_2^C$ where $M_1^C= M^{p_1}_1\sqcup\dots \sqcup M^{p_n}_2$ and $M_2^C=M^{p_1}_2\sqcup\dots\sqcup M^{p_n}_2$. For every~$q\in M^C$, we have already defined a splitting 
$T_{q} M^C=F^s_q\oplus F^u _q$ in section~\ref{s.Dulac} (recall that $M^C=(M^{p_1}_1\cup\dots\cup M^{p_n}_1)\cup (M^{p_1}_2\cup\dots\cup M^{p_n}_2)$ and observe that $q$ is not in $(M^{p_1}_1\cup\dots\cup M^{p_n}_1)\cap (M^{p_1}_2\cup\dots\cup M^{p_n}_2)$). It remains to prove that $\Phi^C$ satisfies~\eqref{e.contraction} and~\eqref{e.dilatation} with respect to these splitting. For this purpose, we will use the decomposition of $\Phi^C$ as a product~:
$$\Phi^C=\Theta^C\circ\Delta^C.$$
The behavior of the derivative of ``Dulac map" $\Delta^C$ was already studied in section~\ref{s.Dulac}~; more precisely, we can rephrase propositions~\ref{p.hyperbolicity-Phi} and~\ref{p.hyperbolicity-Phi'} as follows~:

\begin{proposition}
\label{p.hyperbolicity-Phi-total}
The map $\Delta^C:M^C\to N^C$ is $C^1$. Moreover, for every $q\in \wC\cap M^C$,  the derivative of $D\Delta^C(q):T_q M\to T_{\Delta^C(q)}N$ of the map $\Phi$ at $q$ satisfies~:
\begin{itemize}
\item $D\Delta^C(q).v=0$ for every vector $v\in F^s_q$~;
\item $D\Delta^C(q)$ maps $F^u_q$ on $G^u_{\Delta^C(q)}$, and $\|D\Delta^C(q).v\|_{g_p}=\|v\|_{g_p}$ for every vector $v$ in $F^u_q$.
\end{itemize}
\end{proposition}

It remains to study the behavior of the derivative of the ``transition map" $\Theta^C:N^C\to M^C$. We recall that $\Theta^C(q)$ is well-defined only if the forward orbit of $q$ intersects $M^C$. So our first task is to show that $\Theta^C$ is well-defined at least on a neighborhood of $\wC\cap N^C$ in $N^C$.

\begin{proposition}
\label{p.hyperbolicity-Psi}
There exists a neighborhood $\cV$ of $\wC\cap N^C$ in $N^C$, such that, for every $q\in\cV$, the orbit of $q$ intersects $M^C$ after some time $t(q)$ which depends in a $C^1$ way on $q$. The map $\Theta^C$ is well-defined and $C^1$ on $\cV$. Moreover,  there exists $\nu>1$ such that, for every $q\in \wC\cap N^C$,  the derivative $D\Theta^C(q):T_q N^C\to T_{\Theta^C(q)} M^C$ satisfies
\begin{itemize}
\item $D\Theta^C(q).G^u_q=F^u_{\Theta^C(q)}$ and  $\|D\Theta^C(q).v\|_g\geq \nu \|v\|_g$ for every $v\in G^u_q$.
\end{itemize}
\end{proposition}

\begin{proof}
Consider a point $q\in \wC\cap N^C$. By proposition~\ref{p.orbits-of-C}, the orbit of $q$ intersects $M^C$ at some point $r\in \wC\cap M^C$. Now, recall that~:
\begin{itemize}
 \item $N$, $M^C_1$, $M^C_2$ are $C^1$ hypersurfaces with boundary that are transversal to the orbits of $X_{\cB}$~;
 \item $V$ was chosen so that $C$ is contained in the interior of $V$. This implies that $\wC$ does not intersect neither the boundary of the hypersurface $N$, nor the boundary of hypersurface $M^C_1$ and $M^C_2$. It follows that $\wC$ does not intersect $M^C_1\cap M^C_2$. Hence, $q$ is in the interior of $N$, and $r$ is in the interior of $M^C_1$ or $M^C_2$. 
 \end{itemize}
This implies the existence of a neighborhood $\cV_q$ of $q$ in $N$ such that, for every $q'\in \cV$, the forward orbit of $q'$ intersects $M^C$ after some time $t(q')$ which depends in a $C^1$ way on $q'$. By definition of $\Theta$, for every $q'\in \cV_q$, we have $\Theta(q'):=X_{\cB}^{t(q')}(q')$. In particular, $\Theta$ is well-defined and $C^1$ on $\cV_q$. This proves the two first statements of the proposition.

Since $\cB_{\2}$ is invariant under the flow of $X_{\cB}$, the map $\Theta^C$ maps $\cB_{\2}\cap N^C$ on $\cB_{\2}\cap M^C$. Now, observe that, for every $q\in\cB_{\2}\cap N^C$, the direction $G^u_q$ is nothing but the tangent space of $\cB_{\2}\cap N^C$ at $q$, and the direction $F^u_{\Theta^C(q)}$ is nothing but the tangent space of $\cB_{\2}\cap M^C$ at $\Theta^C(q)$. This shows that $d\Theta^C(q)$ maps $G^u_g$ on $F^u_{\Theta^C(q)}$ for every $q\in \wC\cap N^C$.

We are left to prove the existence of a constant $\nu>1$ such that $\|D\Theta^C(q).v\|_g\geq \nu \|v\|_g$ for every $q\in \wC\cap N^C$ and every $v\in G^u_q$. For this purpose, we will use the maps 
$$\alpha : \cB_{\2}\cap N^C\longrightarrow \cK\cap V^C \quad \mbox{ and }\quad \omega : \cB_{\2}\cap M^C\longrightarrow \cK\cap V^C.$$
We recall that $\alpha$ maps a point $r\in\cB_{\2}\cap N^C$ to the $\alpha$-limit point of the orbit of $r$, and that $\omega$ maps a point $s\in\cB_{\2}\cap M^C$ to the $\omega$-limit point of the orbit of $s$ (see section~\ref{s.Dulac}). We also recall that $\alpha$ is a $C^1$ local isometry for the metrics induced by $g$ on $\cB_{\2}\cap N^C$ and $\cK\cap V^C$, and that $\omega$ is a $C^1$ local isometry for the metrics induced by $g$ on $\cB_{\2}\cap M^C$ and $\cK\cap V^C$ (proposition~\ref{p.isometries}). Finally, we observe that, for $r\in \cB_{\2}\cap N^C$, 
$$\omega(\Theta^C(r))=\omega(r)=f(\alpha(r)).$$ 
The first equality is due to the fact that $\Theta^C(r)$ and $r$ are on the same orbit~; the second one is an immediate consequence of the definition of the Kasner map $f$. This shows that the last statement of proposition~\ref{p.hyperbolicity-Psi} is equivalent to the following statement about the Kasner map~:  there exists a constant $\nu>1$ such that, for every $p\in C$ and every $v\in T_p \cK$,  one has $\|Df(p).v\|_g\geq \nu.|v|_g$. 

This last statement is an immediate consequence of the elementary properties of the Kasner map, and of our choice of the riemannian metric $g^C$. Indeed, the riemannian metric $g^C$ was chosen so that it induces the same metric on $\cK\cap V^C$ as the euclidean metric $h=(d\Sigma_1)^2+(d\Sigma_2)^2+(d\Sigma_3)^2+(dN_1)^2+(dN_2)^2+(dN_3)^2$ (see the end of subsection~\ref{ss.choice-neighborhood} and the end of section~\ref{s.construction-Poincare-section}). And, as we already mentionned in the introduction, since $C$ is a compact subset of the Kasner circle $\cK$ which does not contain any of the three special points $T_1,T_2,T_3$, there exists a constant $\nu^C>1$ such that, for every $q\in C$ and every $v\in T_p \cK$, we have $\|Df(p).v\|_h\geq \nu^C\|v\|$, where $\|\cdot\|_h$ denotes the metric induced on $\cK$ by the euclidean metric $h$.
\end{proof}

Let $\cU:=\Phi^{-1}(\cV)$. Clearly, $\cU$ is a neighborhood of $\wC\cap M^C$ in $M^C$. Combining propositions~\ref{p.hyperbolicity-Phi-total} and~\ref{p.hyperbolicity-Psi}, one immediately gets~:

\begin{proposition}
\label{p.hyperbolicity-Theta}
The Poincar\'e map $\Phi^C$ is well-defined and $C^1$ on $\cU$.  The compact set $\wC\cap M^C$ is a hyperbolic set for $\Phi$. More precisely, there exists a constant $\nu\in (0,1)$ such that, for every $q\in \wC\cap M^C$, 
\begin{itemize}
\item $d\Phi^C(q).v=0$ for every $v\in F^s_q$,
\item $d\Phi^C(q).F^u_q=F^u_{\Phi^C(q)}$, and $\|(d\Phi^C(q))^{-1}.v\|_g<\nu \|v\|_g$ for every $v\in F^u_{\Phi^C(q)}$.
\end{itemize}
\end{proposition}

Proposition~\ref{p.hyperbolicity-Theta} shows that the map $\Phi^C:M^C\to M^C$ and the set $\wC\cap M^C$ satisfy the hypotheses of the stable manifold theorem~\ref{t.stable-manifold} (for any contraction rate $\mu>0$). This shows the existence of local stable manifold, with respect to the map $\Phi^C$, for the points of $\wC\cap M^C$~:

\begin{theorem}
\label{t.stable-manifold-Theta}
For every $\epsilon$ small enough, for every $q\in \wC\cap M^C$, the set 
$$W^s_\epsilon(\Phi^C,q):=\{r\in M^C \mid \mbox{dist}_g\left(\left(\Phi^C\right)^n(r),\left(\Phi^C\right)^n(q)\right)\leq \epsilon \mbox{ for every }n\geq 0\}$$
is a $C^1$-embedded disc of dimension $2$ in $M^C$, tangent to $F^s_q$ at $q$, depending continuously on $q$ in the $C^1$ topology. Moreover, for every constant $\mu>0$, there exists another constant $K$ such that, for every $q\in\wC\cap M^C$, for every $r\in W^s_\epsilon(\Phi^C,q)$ and every $n\geq 0$
$$\mbox{dist}_g\left(\left(\Phi^C\right)^n(r),\left(\Phi^C\right)^n(q)\right)\leq K \mu^n.$$
\end{theorem}


\section{Stable manifolds for the Wainwright-Hsu vector field: proof of theorem~\ref{t.main}}
\label{s.end-proof-main}

We are left to prove that theorem~\ref{t.stable-manifold-Theta} implies our main theorem~\ref{t.main}. 

Consider a point $q\in C$. Then heteroclinic orbit $\cO_{q,f(q)}$ intersects the ``Poincar\'e section" $M^C$ at one and only one point, that we denote by $\bar q$. Note that $\bar q\in\wC\cap M^C$. The set $W^s_\epsilon(\Phi^C,\bar q)$ defined in the statement of theorem~\ref{t.stable-manifold-Theta} is a $C^1$-embedded two-dimensional disc in the three-dimensional hypersurface with boundary $M^C$. This disc is tangent to $F^s_{\bar q}$ at $\bar q$. Since the two-dimensional submanifold $\cB_{\7}\cup \cB_{\2}$ is not tangent to $F^s_{\bar q}$ at $\bar q$, this implies that  $W^s_\epsilon(\Phi^C,\bar q)\cap\cB_{\9}$ contains a $C^1$-embedded two-dimensional  disc in $M^C$. Moreover, this disc depends continuously on $\bar q$. 

\begin{proposition}
\label{p.from-return-map-to-flow}
For every point $r$ in $W^s_\epsilon\left(\Phi^C,\bar q\right)\cap\cB_{\9}$:
\begin{enumerate}
\item \label{claim.1}there is an increasing sequence of times $(t_n)_{n\geq 0}$ such that
$\displaystyle \mbox{dist}_{g}(X_{\cB}^{t_n}(r),f^n(q))\mathop{\longrightarrow}_{n\to\infty} 0;$
\item \label{claim.2} the Hausdorff distance between the piece of orbit $\{X_{\cB}^t(r)\;;\; t_n\leq t\leq t_{n+1}\}$ and the heteroclinic orbit $\cO_{f^n(q),f^{n+1}(q)}$ tends to $0$ when $n$ goes to $+\infty$.
\end{enumerate}
\end{proposition}

\begin{proof}
We first prove item~\ref{claim.1}. According to theorem~\ref{t.stable-manifold-Theta}, we have
\begin{equation}
\mbox{dist}_g\left(\left(\Phi^C\right)^n(r),\left(\Phi^C\right)^n(\bar q)\right)\mathop{\longrightarrow}_{n\to\infty} 0.
\end{equation}
Together with the continuity of the flow of $X_{\cB}$, this shows the existence of a increasing sequence of real numbers $(\tau_n)_{n\geq 0}$ such that 
\begin{equation}
\mbox{dist}_g\left(X_{\cB}^{\tau_n}\left(\left(\Phi^C\right)^n(r)\right),\omega\left(\left(\Phi^C\right)^n(\bar q)\right)\right)\mathop{\longrightarrow}_{n\to\infty} 0.
\end{equation}
Since $r\in\cB_{\9}$, there exists an increasing sequence of times $(\alpha_n)_{n\geq 0}$ such that, for  every $n\geq 0$,
\begin{equation}
\left(\Phi^C\right)^n(r)=X_{\cB}^{\alpha_n}(r).
\end{equation}
Since $\bar q\in\cB_{\2}$, we have, for every $n\geq 0$,
\begin{equation}
\omega\left(\left(\Phi^C\right)^n(\bar q)\right)=f^n(\omega(\bar q))=f^{n+1}(q).
\end{equation}
For every $n\geq 0$, let $t_{n+1}:=\tau_n+\alpha_n$. Then $(t_n)_{n\geq 0}$ is an increasing sequence, and 
\begin{equation}
\label{e.convergence}
\mbox{dist}_{g}\left(X_{\cB}^{t_{n+1}}(r),f^{n+1}(q)\right)\mathop{\longrightarrow}_{n\to\infty} 0.
\end{equation}
This completes the proof of item~\ref{claim.1}.

\medskip

To prove item~\ref{claim.2}, we decompose the piece of orbit $\{X_{\cB}^t(r)\;;\; t_n\leq t\leq t_{n+1}\}$ into three sub-pieces:
\begin{itemize}
\item First, the piece of orbit going from $X_{\cB}^{t_n}(r_0)$ to $\Delta^C\left(\left(\Phi^C\right)^n(r)\right)$, contained in $V$. Formula~\eqref{e.linear-flow}, together with~\eqref{e.convergence} shows that, for $n$ large, this piece of orbit is close to the heteroclinic orbit $\cO_{f^n(q),f^{n+1}(q)}$. In particular, for $n$ large, the point $\Delta^C\left(\left(\Phi^C\right)^n(r)\right)$ is close to the heteroclinic orbit $\cO_{f^n(q),f^{n+1}(q)}$.
\item Then, a piece of orbit  going from $\Delta^C\left(\left(\Phi^C\right)^n(r)\right)$ to $\left(\Phi^C\right)^{n+1}(r)$, contained in $\cB\setminus V$. For $n$ large, this piece of orbit  is close to the heteroclinic orbit $\cO_{f^n(q),f^{n+1}(q)}$. Indeed, for $n$ large, the point $\Delta^C\left(\left(\Phi^C\right)^n(r_0)\right)$ is close to the heteroclinic orbit $\cO_{f^n(q),f^{n+1}(q)}$, and if we write $\left(\Phi^C\right)^{n+1}(r)=X_{\cB}^{t(\Delta^C((\Phi^C)^n(r)))}\left(\Delta^C\left(\left(\Phi^C\right)^n(r)\right)\right)$, then $t\left(\Delta^C\left(\left(\Phi^C\right)^n(r)\right)\right)$ depends continuously on $\Delta^C\left(\left(\Phi^C\right)^n(r)\right)$ (proposition~\ref{p.hyperbolicity-Psi}) and thus is uniformly bounded. 
\item Finally, a piece of orbit going from $\left(\Phi^C\right)^{n+1}(r)$ to $X_{\cB}^{t_{n+1}}(r)$, contained in $V$. Formula~\eqref{e.linear-flow} together with~\eqref{e.convergence} show that, for $n$ large, this piece of orbit is close to the heteroclinic orbit $\cO_{f^n(q),f^{n+1}(q)}$. 
\end{itemize}
This completes the proof of item~\ref{claim.2}
\end{proof}

\begin{corollary}
\label{c.from-return-map-to-flow-2}
$$W^s(q)=\bigcup_{t\geq 0}\bigcup_{n\geq 0} X^{-t}\left(W^s_\epsilon\left(\Phi^C,\left(\Phi^C\right)^n(\bar q)\right)\cap\cB_{\9}\right).$$
\end{corollary}

\begin{proof}
The inclusion of the set on the right hand side in $W^s(q)$ follows from proposition~\ref{p.from-return-map-to-flow}. The inclusion of $W^s(q)$ in the set on the right hand side is an immediate consequence of the definition of~$W^s(q)$.
\end{proof}

We can now complete the proof of our main theorem.

\begin{proof}[Proof of theorem~\ref{t.main}] 
Fix $\eta>0$, and we set 
\begin{equation}
D^s_{\9}(q):=\bigcup_{-\eta\leq t\leq \eta} X_{\cB}^t\left(W^s_{\epsilon}\left(\Phi^C,\bar q\right) \cap \cB_{\9}\right).
\end{equation}
According to corollary~\ref{c.from-return-map-to-flow-2}, $D^s_{\9}(q)$ is contained in $W^s(q)$. Since $W^s_{\epsilon}(\Phi^C,\bar q))\cap\cB_{\9}$ is a $C^1$-embedded 2-dimensional disc in $M^C$, and since the orbits of $X_{\cB}$ are transversal to $M^C$, we get that $D^s_{\9}(q)$ is a $C^1$-embedded 3-dimensional disc in $\cB_{\9}$.  Since $\bar q$ depends continuously on $q$ (proposition~\ref{p.isometries}), and since $W^s_\epsilon(\Phi^C,\bar q)$ depends continuously on $\bar q$ (theorem~\ref{t.stable-manifold-Theta}), the disc $D^s_{\9}(q)$ depends continuously on $q$.
\end{proof}

\begin{remark}
The fact that $W^s(q)$ is an $C^1$ injectively immersed open disc which depends continuousluy on $q$ (remark~\ref{r.global-stable-manifold}) almost follows from the same arguments. More precisely, theorem~\ref{t.stable-manifold-Theta}, corollary~\ref{c.from-return-map-to-flow-2} and the transversality of $M^C$ to the orbits of $X$ show that $W^s(q)$ is an increasing union of $C^1$ embedded closed discs which depend continuously on $q$. The only thing which remains to shows is that this increasing union of closed discs is an open disc; this is actually a consequence of the fact that the orbit of $q$ under the Kasner map $f$ is not periodic. 
\end{remark}


\section[Existence of closed invariant aperiodic subsets of the Kasner circle]{Existence of closed forward-invariant aperiodic subsets of the Kasner circle: proof of proposition~\ref{p.existence-aperiodic}}
\label{s.existence-aperiodic}

The purpose of this section is to prove proposition~\ref{p.existence-aperiodic}. This proposition should be quite obvious for people with some culture in dynamical systems. Indeed, the Kasner map $f:\cK\to\cK$ is a degree $-2$ map of the circle $\cK$. This implies the existence of a continuous degree~1 map $\eta:\cK\to\RR/\ZZ$ such that $\eta\circ f=m_{-2}\circ \phi$ where $m:\RR/\ZZ\to\RR/\ZZ$ is defined by $m(\theta)=-2\theta$.  Moreover, the Kasner map $f$ is expansive~: the norm of the derivative of $f$ (calculated with respect to the metric induced on $\cK$ by the rimennian metric $h$) is strictly bigger than $1$, except at the three special points $T_1,T_2,T_3$ (where it is equal to $1$). This implies that the map $\eta:\cK\to\RR/\ZZ$ defined above is one-to-one, \emph{that is} $f:\cK\to \cK$ is topologically conjugated to the map $m:\theta\mapsto -2\theta$. Finally, it is well-known by experts that, for $|k|\geq 2$, the union of all compact subsets of $\RR/\ZZ$ which are aperiodic for the map $\theta\mapsto k\theta$ is dense in $\RR/\ZZ$. We now give a more detailed proof of the proposition for the readers who may not necessarily be familiar with low-dimensional dynamics.

\begin{proof}
Recall that $T_1,T_2,T_3$ are the three Taub points on the Kasner circle $\cK$. Let $I_1,I_2,I_3$ be the closures of the three connected components of $\cK\setminus\{T_1,T_2,T_3\}$, the notations being chosen so that $T_1$ is not one end of $I_1$, $T_2$ is not one end of $I_2$, and $T_3$ is not an end of $I_3$. 

We consider the set $\Sigma:=\{1,2,3\}^\NN$ endowed with the product topology, and the shift map $\sigma:\Sigma\to\Sigma$ defined by $\sigma(a_0,a_1,a_2,\dots)_i=(a_1,a_2,a_3,\dots)$ (in other words, if   $\ba=(a_i)_{i\in\NN}\in\Sigma$, 
then $(\sigma(\ba))_i=a_{i+1}$). Let $\Sigma_0$ be the subset of $\Sigma$ defined as follows~: 
$$\Sigma_0=\{\ba=(a_i)_{i\in\NN}\in\Sigma \mbox{ such that }a_{i+1}\neq a_i\mbox{ for every }i\}.$$
Note that $\Sigma_0$ is $\sigma$-invariant. We will construct a continuous ``almost one-to-one" map $h: \Sigma_0\to \cK$ such that $h\circ\sigma=f\circ h$. 

\medskip

\noindent \textit{Claim. For each sequence $\ba:=(a_i)_{i\in\NN}$ in $\Sigma_0$, there exists a unique point $p\in\cK$ such that $f^i(p)\in I_{a_i}$ for every $i\geq 0$.}

In order to prove the existence of $p$, one just needs to notice that the image under $f$ of each of the intervals $I_1$, $I_2$, $I_3$ is the union of the two other intervals. This implies that the intersection $\bigcap_{i=0}^N f^{-i}(I_{a_i})$ is non-empty for every $N$, and therefore, that the intersection $\bigcap_{i\in\NN} f^{-i}(I_{a_i})$ is non-empty. The existence of $p$ follows. In order to prove the uniqueness of $p$, observe that: for every $\epsilon>0$, there exists $\nu(\epsilon)>1$ such that $\||Df(p)\||_h\geq \nu(\epsilon)$ for every $x\in\cK$ such that $\mbox{dist}(x,T_i)>\epsilon$ for $i=1,2,3$. Hence, if $p\neq p'$ were two points such that $f^i(p)\in I_{a_i}$ and $f^i(p')\in I_{a_i}$ for every $i\geq 0$, then one would have $\mbox{dist}(f^{i}(p),f^{i}(p'))\rightarrow\infty$~; this is absurd since the lengths of $I_1$, $I_2$ and $I_3$ are finite. Hence there is at most one point $p$ in $\bigcap_{i\in\NN} f^{-i}(I_{a_i})$. This completes the proof of the claim. 

\bigskip

Now, we consider the map $h:\Sigma_0\to\cK$ which maps a sequence $\ba:=(a_i)_{i\in\NN}$ to the unique point $p\in\cK$ such that $f^i(p)\in I_{a_i}$ for every $i\geq 0$. This map $h$ obviously satisfies $h\circ\sigma=f\circ h$. It is continuous (this is an immediate consequence of the continuity of $f$) and onto (because the image under $f$ of one of the intervals $I_1,I_2,I_3$ is contained in the union of the two others). It is \emph{not} one-to-one. For example, the Taub point $T_3$ has two pre-images under $h$~:  the sequences $(1,2,1,2,1,2,\dots)$ and $(2,1,2,1,2,1,\dots)$. More generally, every point $x\in\cK$ such that $f^{i_0}(p)$ is a Taub point for some integer $i_0\geq 0$ (and thus $f^i(p)=f^{i_0}(p)$ for every $i\geq i_0$) has two pre-images under $h$, and these two pre-images are preperiodic for $\sigma$. This is the only lack of injectivity of $h$~: if $p\in\cK$ has not a single pre-image under $h$, then there exists $i_0\geq 0$ such that $f^{i_0}(p)$ is a Taub point (this follows from the fact that the intersection between two of the three intervals $I_1,I_2,I_3$ is reduced to a Taub point).

It follows that the image under $h$ of a closed $\sigma$-invariant aperiodic subset of $\Sigma_0$ is a closed  forward-invariant aperiodic subset of $\cK$. So we are left to prove that the union of all the closed $\sigma$-invariant and aperiodic subsets of $\Sigma_0$ is in dense in~$\Sigma_0$.

A element $\ba=(a_i)_{i\geq 0}$ of $\Sigma_0$ is said to be \emph{square-free} if it does not contain the same word repeated twice~: for every $i_0\geq 0$ and every $\ell>0$, the word $a_{i_0}\dots a_{i_0+\ell-1}$ is different from the word $a_{i_0+\ell}\dots a_{i_0+2\ell-1}$. It is well-known that there exist square-free elements in $\Sigma_0$ (such an element may be easily deduced from the well-known Prouhet-Thue-Morse sequence, see for example~\cite[corollary~1]{AlloucheShallit1999}). Now let $\ba=(a_i)_{i\geq 0}\in\Sigma_0$ be square-free, then the $\sigma$-orbit of $\ba$ does not accumulates on any periodic $\sigma$-orbit, hence the closure of the $\sigma$-orbit of $\ba$ is a closed  forward-invariant aperiodic subset of $\Sigma_0$. Moreover the same is true, if one replaces $\ba$ by a sequence $\ba'\in\Sigma_0$ which has the same tail as $\ba$ (\emph{i.e.} there exists $i_0$ such that $a'_i=a_i$ for $i\geq i_0$). The set of all sequences $\ba'$ which has the same tail has $\ba$ is obviously dense in $\Sigma_0$. Hence the union of all closed $\sigma$-invariant aperiodic subsets of $\Sigma_0$ is dense in $\Sigma_0$. As explained above, the proposition follows. 
\end{proof}


\end{document}